\documentclass[11pt, onecolumn, letter]{IEEEtran}
\usepackage{latexsym}
\usepackage{amssymb}
\usepackage{amsmath}
\usepackage{amsthm}
\usepackage[dvips]{graphicx}
\usepackage{multirow}

\usepackage{color}

\newcommand{\san}[1]{\mathsf{#1}}
\newcommand{\rom}[1]{\mathrm{#1}}

\newcommand{\Pj}{\rom{P}_{\rom{j}}}
\newcommand{\Pjs}{\rom{P}_{\rom{js}}}
\newcommand{\Ps}{\rom{P}_{\rom{s}}}

\newcommand{\K}{\rom{K}}

\newtheorem{definition}{Definition}
\newtheorem{proposition}[definition]{Proposition}
\newtheorem{theorem}{Theorem}
\newtheorem{corollary}{Corollary}
\newtheorem{remark}[definition]{Remark}
\newtheorem{lemma}{Lemma}

 \newenvironment{proofof}[1]{\vspace*{5mm} \par \noindent
     \quad{\it Proof of #1:\hspace{2mm}}}{\endproof
}

\def\Label#1{\label{#1}\ [\ \text{#1}\ ]\ }
\def\Label{\label}

\begin{document}

\title{Finite-Length Bounds for Joint Source-Channel Coding with Markovian Source and Additive Channel Noise to Achieve Large and Moderate Deviation Bounds}

\author{
 \IEEEauthorblockN{Ryo Yaguchi$^{a}$ and Masahito Hayashi$^{a,b}$}\\
  \IEEEauthorblockA{$~^{a}$Graduate School of Mathematics, Nagoya University \\
$^b$Centre for Quantum Technologies, National University of Singapore \\
    Email: {yaguchi.riyou@c.mbox.nagoya-u.ac.jp \& masahito@math.nagoya-u.ac.jp} }
} 

\maketitle

\begin{abstract}
We derive novel upper and lower finite-length bounds of the error probability in joint source-channel coding when the source obeys an ergodic Markov process and the channel is a Markovian additive channel or a Markovian conditional additive channel. These bounds are tight in the large and moderate deviation regimes.
\end{abstract}

\begin{IEEEkeywords} 
Markov chain, 
joint source-channel coding, 
finite-length analysis,
large deviation,
moderate deviation
\end{IEEEkeywords}

\section{Introduction}\Label{S1}
Shannon theoretic information theory originally 
focuses on the asymptotic performance.
Since the block length of any real code is finite,
analysis with finite-blocklength is more important in a practical setting.
Although the tight analysis is possible in the asymptotic regime,
it is almost impossible in the finite-length regime.
Hence, we usually take a strategy to find good upper and lower bounds of the decoding error probability
in the finite-length regime.
Since lower and upper bounds are not unique,
we need several requirements for the bounds to clarify their goodness.
One is the asymptotic tightness.
That is, we impose the first condition that the limit of the bound attains one of the following regimes:
(1) Second order, 
(2) Moderate deviation, and
(3) Large deviation.

To satisfy the above requirement, one may use the minimum value with respect to so many parameters.
If the calculation complexity for the bound is too huge, 
it cannot be used in a practical use because we cannot calculate the bound.
To estimate the optimal performance for a given blocklength $n$,
we need to impose the second condition that its calculation complexity is not so large, e.g., 
$O(1)$, $O(n)$, or $O(n \log n)$.

Usually, the channel coding is discussed with the message subject to the uniform distribution.
However, in the real communication, the message is not necessarily subject to the uniform distribution.
To resolve this problem, we often consider the channel coding with 
the message subject to the non-uniform distribution.
Such a problem is called source-channel joint coding and has been actively studied by several researchers \cite{Csis,CVFKM,ZA,DAY,KV,VSM}.

As a simple case, we often assume that the message is subject to an independent and identical distribution.
In this case, the capacity is given as the ratio of
the conventional channel capacity to the entropy of the message.
Recently, Wang-Ingber-Kochman \cite{DAY} and Kostina-Verd\'{u} \cite{KV} discussed the second-order coefficient in this problem.
In the same setting, the papers \cite{Csis,ZA2,CVFM,CVFKM} derived the exponential decay rate of the minimum decoding error probability
when the information source is subject to an independent and identical distribution and the channel is a discrete memoryless channel.
Now, we focus on the case when the information source obeys a Markovian process and 
the channel is affected by additive noise that simply obeys Markovian process.
In this setting,
the paper \cite{ZA} derived a lower bound of the exponential decreasing rate of the minimum decoding error probability,
and 
the paper \cite{VSM} derived the moderate deviation of the same error probability.
That is, their direct part \cite{VSM} follows from the idea of the paper \cite{ZA},
and their converse part \cite{VSM} follows from their new idea.
However, they did not derived a finite-length bound without polynomial overhead.

The recent paper \cite{HW} discussed the channel coding when 
the distribution of the additive noise in the channel is decided by the channel state,
and the channel state is observed by the receiver and is subject to Markovian process. 
Such a channel is called a conditional additive channel. 
For example, Gilbert-Elliot channel with state-information available at the receiver is written as a special case of the former setting, but cannot be written as a special case of the latter setting.
Hence, it is needed to treat a conditional additive channel to adopt a more realistic situation.
In this paper, 
we focus on two kinds of assumptions 
(Assumptions 1 and 2)
for such generalized additive noise channels. 
Under these assumptions for channels,
we address joint source-channel coding with Markovian source and conditional additive channel noise.

As summarized in Tables \ref{comparison} and \ref{comparison2},
the contribution of this paper is the following two points.
One is 
to derive large and moderate deviation bounds under the above general setting, which are the generalizations of the results by the papers \cite{ZA,VSM}.
The other is to derive upper and lower bounds 
with computable forms of the decoding error probability that match in the large deviation regime in the above general setting
while the papers \cite{ZA,VSM} did not give finite-length bounds in a computable form in our sense.

\begin{table*}
\caption{Comparison of upper bounds of decoding error probability for 
joint source-channel coding in the additive channel noise case}
\label{comparison}
\begin{center}
\begin{tabular}{|c|c|c|c|c|c|c|}
\hline
& Tight  & Finite bound  & Markov  & \multirow{2}{*}{Markov} & Markov & \multirow{2}{*}{Linear}  \\
& exponent & without  & channel &   & conditional   
&  \\
& in IID case & polynomial factor & noise & source  & additive channel 
& code \\
\hline
\cite{KV} 
 & No & Yes & No  & No & No & No \\
\hline
\cite{Csis}
& Yes & No   & No  & No & No & No \\
\hline
\cite{ZA2,CVFM,CVFKM}
& Yes & Yes   & No  & No & No & No \\
\hline
\cite{ZA} +\cite{VSM}
& Yes & No & Yes & Yes & No & Yes \\
\hline
Proposed 
& Yes & Yes & Yes & Yes & Yes & Yes \\
\hline
\end{tabular}
\end{center}
Tight exponent in IID case shows the tightness over the critical rate. 
The paper \cite{KV} derived a finite bound without 
polynomial factor.
However, they did not discuss the calculation complexity.
\end{table*}

\begin{table*}
\caption{Comparison of lower bounds of decoding error probability for 
joint source-channel coding in the additive channel noise case}
\label{comparison2}
\begin{center}
\begin{tabular}{|c|c|c|c|c|c|}
\hline
& Tight  & Finite bound  & Markov  & \multirow{2}{*}{Markov} & Markov  \\
& exponent & without  & channel &   & conditional   
  \\
& in IID case & polynomial factor & noise & source  & additive channel  \\
\hline
\cite{KV} 
 & No & Yes & No  & No & No \\
\hline
\cite{Csis}
& Yes & No   & No  & No & No \\
\hline
\cite{VSM}
& Yes & No   & Yes  & Yes & No  \\
\hline
Proposed 
& Yes & Yes & Yes & Yes & Yes \\
\hline
\end{tabular}
\end{center}
The papers \cite{ZA2,ZA,CVFM,CVFKM} did not derive an efficient lower bound 
of the decoding error probability when the error goes to zero exponentially while the paper 
\cite{ZA2} discussed the relation of the obtained upper bound with the lower bound by \cite{Csis}.
For the relation with \cite{VSM}, see Remark \ref{R1}.
\end{table*}

\begin{table}[htpb]
  \caption{Summary of results.}\Label{T3}
\begin{center}
\begin{tabular}{|c|c|c|c|c|c|}
\hline
	 		& Channel &Finite & LD & MD &  Complexity\\
\hline
\multirow{2}{*}{Direct} & Ass. 1& Theorem \ref{f1d} & Theorem \ref{T5} & Theorem \ref{md}  & ${\cal O}(1)$ \\
\cline{2-4}\cline{6-6}
 & Ass. 2 & Theorem \ref{f2d} & Theorem \ref{ld2d} (Tight) & (Tight) &  ${\cal O}(1)$\\ 
\hline
\multirow{2}{*}{Converse}& Ass. 1& Theorem \ref{f1c} & Theorem \ref{T6}  & Theorem \ref{md}  &  
${\cal O}(1)$\\
\cline{2-4}\cline{6-6}
 & Ass. 2 & Theorem \ref{f2c} & Theorem \ref{ld2c} (Tight) &  (Tight)	&  ${\cal O}(1)$\\ 
\hline
\end{tabular}
\end{center}
Assumption 1 contains Assumption 2.
``Finite'', ``LD'', and ``MD'' express the finite-length bound,
the large deviation bound, and the moderate deviation bound, respectively.
\end{table}

The remaining part of this paper is organized as follows.
In Section \ref{S2}, we prepare several information quantities for Markovian process.
Section \ref{S3} prepares several useful functions for finite-length analysis.
Section \ref{S4} explains several useful lemmas under the single shot setting.
Section \ref{S4} shows our main results, i.e., our finite-length bounds and 
large and moderate deviation bounds.
Section \ref{S6} gives 
our numerical analysis based on our finite-length bounds.
Table \ref{T3} explain the summary of our results.

\section{Information Measures for two  terminals}\Label{S2}
In this section, we introduce some information measures and their properties will be used in latter sections. 
\subsection{Information measures for single-shot setting}
Since this paper addresses finite-length setting and the large deviation analysis, we need the conditional R\'{e}nyi entropy. When the joint distribution is given to be $ P_{ XY } $ the conditional R\'{e}nyi entropy relative to $ Q_Y $ is given as
\begin{align}
&H_{1-\theta}(P_{XY}|Q_Y) := \frac{1}{\theta}\log \sum_{x, y} P_{XY}(x, y)^{1-\theta}Q_Y(y)^{\theta}. 
\end{align}

Dependently of the choice fo the distribution $Q_Y$, we have the upper and lower types of conditional R\'{e}nyi entropy:
\begin{align}
H_{1-\theta}^{\downarrow}(X|Y)& := H_{1-\theta}(P_{XY}|P_Y), \\
H_{1-\theta}^{\uparrow}(X|Y) &:= H_{1-\theta}(P_{XY}|P_Y^{1-\theta}), 
\end{align}
where
\begin{align}
P_Y^{1-\theta}(y) := \frac{[\sum_{x} P_{XY}(x, y)^{1-\theta}]^{\frac{1}{1-\theta}}}{\sum_{y'}[\sum_{x} P_{XY}(x, y')^{1-\theta}]^{\frac{1}{1-\theta}}}. 
\end{align}
To connect these two types of conditional R\'{e}nyi entropy, we often focus on the following type of conditional R\'{e}nyi entropy

\begin{align}
H_{1-\theta, 1-\theta'}(X|Y) := H_{1-\theta}(P_{XY}|P_Y^{1-\theta'}). \Label{renyi3}
\end{align}

For $ P,Q \in {\cal P(X)} $, we define R\'{e}nyi divergence
\begin{align}
D_{1+s}(P||Q)
:=
 \frac{ 1 }{ s } \log \sum_x P(x)^{ 1+s } Q(x)^{ -s }.
\end{align}

Using R\'{e}nyi divergence, we introduce two types of R\'{e}nyi mutual informations
\begin{align}
I_{ 1-s }^{ \downarrow }(X;Y | P_{XY})
:=&
D_{ 1-s } (P_{XY} || P_X \times P_Y) ,\\
I_{1-s}^\uparrow(X;Y|P_{XY} )
:= &-\frac{1-s}{s}\log 
\sum_{y} (\sum_{x} P_X(x) P_{Y|X}(y|x)^{1-s} )^{\frac{1}{1-s}}
\end{align}

\subsection{Information measures for transition matrix}
Since this paper address the Markovian information source, we prepare several information measures given in \cite{HW} for an ergodic and irreducible transition matrix $W = \{W(x, y|x', y') \}_{(x, y), (x', y') \in ({\cal X \times Y} )^2 }$ on $({\cal X} \times {\cal Y})$. For this purpose, we employ two assumptions on transition matrices, which were introduced by the paper \cite {HW}. 

\begin{definition}[Assumption 1 (non-hidden)]
We assume the following condition for a transition matrix $W$:
\begin{align}
\sum_{x} W(x, y|x', y') = W(y|y'),
\end{align}
for every $x' \in {\cal X}$ and $y, y' \in {\cal Y}$. 

When this condition holds, a transition matrix $ W $ is called 
non-hidden (with respect to $ {\cal Y} $).
\end{definition}

\begin{definition}[Assumption 2]
We assume one of the following conditions for a transition 
matrix $ W $:
\begin{enumerate}
	\item for every $  \theta \in (- \infty, 0) $ and 
	$ (y, y') \in {\cal Y} \times {\cal Y} $, 
	\begin{align}
	W_\theta (y|y') = \sum_{x} W(x, y|x', y') ^{ 1 - \theta }. \Label{A2}
	\end{align}
	is well defined, i.e., the right hand side of (\ref{A2}) 
	is independent of $ x' $.
	
	When this condition holds, a transition matrix $ W $ is 
	called strongly non-hidden (with respect to $ {\cal Y} $).
	\item $ |{\cal Y}| = 1 $. 
	
	When this condition holds, a transition 
	matrix $ W $ is called singleton. 
\end{enumerate}

\end{definition}

Assumption 1 is acquired from (\ref{A2}) by substituting $\theta = 0$, so Assumption 2 implies Assumption 1. 
When a transition matrix on $W$ satisfies Assumption 1, 
we define the marginal $W_Y$ by $W_Y(y|y') := \sum_x W(x, y|x', y')$. 
For the transition matrix $T$ on ${\cal Y}$, 
we also define ${\cal Y}^2_T := \{(y, y') : T(y|y') > 0\}$. 
Then, when another transition matrix $V$ on ${\cal Y}$ satisfies 
$ {\cal Y}^2_{W_Y} \subset {\cal Y}^2_V$, we define

\begin{align}
H_{1-\theta}^{W|V}(X|Y) := \frac{1}{\theta} \log{\lambda_{\theta}^{W|V}}, 
\end{align}
where $\lambda_{\theta}^{W|V}$ is the Perron-Frobenius eigenvalue of
\begin{align}
W(x, y|x', y')^{1-\theta}V(y|y')^{\theta}. 
\end{align}

\noindent Then, the lower type of conditional R\'{e}nyi entropy for the transition matrix \cite{HW} is given as

\begin{align}
H_{1-\theta}^{W, \downarrow}(X|Y) := H_{1-\theta}^{W|W_Y}(X|Y). 
\end{align}

\noindent Also, when $ W $ satisfies Assumption 2, the upper type of conditional R\'{e}nyi entropy for the transition matrix \cite{HW} is given as

\begin{align}
H_{1-\theta}^{W, \uparrow}(X|Y) := \max_{V} H_{1-\theta}^{W|V}(X|Y). 
\end{align}

Furthermore, we define the information measure which is counterpart of (\ref{renyi3}). For this purpose, we introduce the following $|{\cal Y}| \times |{\cal Y}|$ matrix:

\begin{align}
N_{\theta, \theta'}(y|y') := W_\theta(y|y')W_{\theta'}(y|y')^{\frac{\theta}{1-\theta'}},
\end{align}
where $W_\theta(y|y')$ is defined in \eqref{A2}.
\noindent Let $\nu_{\theta, \theta'}$ be the Perron-Frobenius eigenvalue of $N_{\theta, \theta'}$. Then, we define the two-parameter conditional R\'{e}nyi entropy \cite{HW} by

\begin{align}
H_{1-\theta, 1-\theta'}^{W}(X|Y) := \frac{1}{\theta}\log{\nu_{\theta, \theta'}}
	-\frac{\theta'}{1-\theta'}H_{1-\theta'}^{W, \uparrow}(X|Y). 
\end{align}

For $ \theta = 0 $, we define the conditional R\'{e}nyi entropy for $ W $ by

\begin{align}
H^W (X|Y)
:=
\lim_{ \theta \to 0 }H_{ 1 - \theta }^{W, \downarrow}(X|Y). \Label{Hw}
\end{align}
Also, we define following quantity.

\begin{align}
V^W( X|Y )
:=
\lim_{ \theta \to 0 }
\frac{2[H_{ 1 - \theta }^{W, \downarrow}(X|Y) - H^W(X|Y)]} { \theta }. \Label{Vw}
\end{align}
According to \cite{HW}, using (\ref{Hw}) and (\ref{Vw}), we obtain the following 
two expansions.

\begin{align}
&H_{ 1 - \theta }^{W, \downarrow}(X|Y)
= H^W (X|Y) + \frac{ \theta }{ 2 } V^W( X|Y ) + o(\theta),\Label{Hdex}\\
&H_{ 1 - \theta }^{W, \uparrow}(X|Y)
= H^W (X|Y) + \frac{ \theta }{ 2 } V^W( X|Y ) + o(\theta)\Label {Huex}
\end{align}
around $ \theta = 0 $.

Under these preparations, we have three lemmas as follows. 

\begin{proposition}\Label{l1}\cite[lemma 9]{HW}
Suppose that a transition matrix W satisfies Assumption 1.  
Let $W_\theta(x, y) := W(x, y|x', y')^{1-\theta}W(y|y')^\theta $ and $\it{v}_{\theta}$ be the eigenvector of $W_{\theta}^T$ with respect to the Perron-Frobenius eigenvalue $\lambda_{\theta}$ such that $\min_{x, y} \it{v}_{\theta}(x, y) = 1$. 
Let $\it{w}_\theta(x, y) = P_{X_1Y_1}(x, y)^{1-\theta}P_{Y_1}(y)^\theta$. Then, we have
\begin{align}
(n-1)\theta H_{1-\theta}^{W, \downarrow}(X|Y) + \underline{\delta}_W(\theta)
\le \theta H_{1-\theta}^{\downarrow}(X^n|Y^n)
\le (n-1) \theta H_{1-\theta}^{W, \downarrow}(X|Y) + \overline{\delta}_W(\theta), \Label{EFB1}
\end{align}
where
\begin{align}
&\overline{\delta}_W(\theta) := \log{\it{v}_{\theta}\cdot\it{w}_\theta}, \\
&\underline{\delta}_W(\theta) 
:= \log{\it{v}_{\theta}\cdot\it{w}_\theta} - \log{\max_{x, y}\it{v}_{\theta}(x, y)}. 
\end{align}
\end{proposition}

\begin{proposition}\Label{l2}\cite[lemma 10]{HW}
Suppose that a transition matrix W satisfies Assumption 2. 
Then, we have
\begin{align}
(n-1)\frac{\theta}{1-\theta} H_{1-\theta}^{W, \uparrow}(X|Y) + \underline{\xi}_W(\theta)
\le \frac{\theta}{1-\theta} H_{1-\theta}^{\uparrow}(X^n|Y^n)
\le (n-1) \frac{\theta}{1-\theta}H_{1-\theta}^{W, \uparrow}(X|Y) + \overline{\xi}_W(\theta), 
\end{align}
where $ \overline{\xi}_W(\theta) $ and $ \underline{\xi}_W(\theta) $ 
is defined as follows: 

For the non-hidden case, 
we define the $|{\cal Y}| \times |{\cal Y}|$ matrix $K_\theta$ so that 
\begin{align}
K_\theta(y|y') := [\sum_x W(x, y|x', y')^{1-\theta}]^{\frac{1}{1-\theta}}, 
\end{align}
and $\it{v}_{\theta}$ be the eigenvector of 
$K_{\theta}^T$ with respect to the Perron-Frobenius eigenvalue 
$\kappa_{\theta}$ such that $\min_{y} \it{v}_{\theta}(y) = 1$. 
Let $\it{w}_\theta$ be the $|{\cal Y}|$-dimensional vector defined by
\begin{align}
\it{w}_\theta(y) = \left[ \sum_x P_{X_1Y_1}(x, y)^{1-\theta} \right]^{\frac{1}{1-\theta}}. 
\end{align}
Then, $ \overline{\xi}_W(\theta) $ and $ \underline{\xi}_W(\theta) $ 
are defined as:
\begin{align}
&\overline{\xi}_W(\theta) := \log{\it{v}_{\theta}\cdot\it{w}_\theta}, \\
&\underline{\xi}_W(\theta) 
:= \log{\it{v}_{\theta}\cdot\it{w}_\theta} - \log{\max_{y}\it{v}_{\theta}(y)}. 
\end{align}

For the singleton case, let $W_\theta(x) := W(x|x')^{1-\theta}$ and $\it{v}_{\theta}$ be the eigenvector of $W_{\theta}^T$ with respect to the Perron-Frobenius eigenvalue $\lambda_{\theta}$ such that $\min_{x} \it{v}_{\theta}(x) = 1$. 
Let $\it{w}_\theta(x) = P_{X_1}(x)^{1-\theta}$.
Then,
$\overline{\xi}_W(\theta)$ and $\underline{\xi}_W(\theta)$ are defined as:
\begin{align}
&\overline{\xi}_W(\theta) := \log{\it{v}_{\theta}\cdot\it{w}_\theta}, \Label{pro1}
\\
&\underline{\xi}_W(\theta) 
:= \log{\it{v}_{\theta}\cdot\it{w}_\theta} - \log{\max_{x}\it{v}_{\theta}(x)}. \Label{pro2}
\end{align}

\end{proposition}

\begin{proposition}\Label{l3}\cite[lemmas 9 and 11]{HW}
Suppose that a transition matrix $W$ satisfies Assumption 2. 
Then, we have
\begin{align}
(n-1)\theta H_{1-\theta, 1-\theta'}^{W}(X|Y) + \underline{\zeta}_W (\theta, \theta')
\le
\theta H_{1-\theta, 1-\theta'}(X^n|Y^n)
\le
(n-1)\theta H_{1-\theta, 1-\theta'}^{W}(X|Y) + \overline{\zeta}_W (\theta, \theta')
\end{align}
where $ \overline{\zeta}_W (\theta, \theta') $ and 
$ \underline{\zeta}_W (\theta, \theta') $ are defined as follows: 

For the non-hidden case with respect to $ {\cal Y} $, 
let $\it{v}_{\theta, \theta'}$ be the eigenvector of $N_{\theta, \theta'}^T$ with respect to the Perron-Frobenius eigenvalue $\nu_{\theta, \theta'}$ such that $\min_y \it{v}_{\theta, \theta'}(y) = 1$. Let $\it{w}_{\theta, \theta'}$ be the $|{\cal Y}|$-dimensional vector defined by
\begin{align}
\it{w}_{\theta, \theta'}(y) := \left[ \sum_x P_{X_1Y_1}(x, y)^{1-\theta} \right]
\left[ \sum_x P_{X_1Y_1}(x, y)^{1-\theta'} \right] ^{\frac{\theta}{1-\theta'}}. 
\end{align}
Then, $ \overline{\zeta}_W (\theta, \theta') $ and 
$ \underline{\zeta}_W (\theta, \theta') $ are defined as:
\begin{align}
\overline{\zeta}_W (\theta, \theta') :=&\log{\it{v}_{\theta, \theta'}\cdot\it{w}_{\theta, \theta'}}
	-\theta \overline{\xi}_W(\theta'),\\
\underline{\zeta}_W (\theta, \theta') :=&\log{\it{v}_{\theta, \theta'}\cdot\it{w}_{\theta, \theta'}}
	-\log{\max_{y}\it{v}_{\theta, \theta'}(y)}-\theta \underline{\xi}_W(\theta'),
\end{align}
for $\theta <0$ and
\begin{align}
\overline{\zeta}_W (\theta, \theta') :=&\log{\it{v}_{\theta, \theta'}\cdot\it{w}_{\theta, \theta'}}
	-\theta \underline{\xi}_W(\theta'), \\
\underline{\zeta}_W (\theta, \theta') :=&\log{\it{v}_{\theta, \theta'}\cdot\it{w}_{\theta, \theta'}}
	-\log{\max_{y}\it{v}_{\theta, \theta'}(y)}-\theta \overline{\xi}_W(\theta'), 
\end{align}
for $\theta >0$. 

For the singleton case, 
we define 
$\overline{\zeta}_W (\theta, \theta')$ and $\underline{\zeta}_W (\theta, \theta')$ 
by \eqref{pro1} and \eqref{pro2} independently of $\theta'$.
\end{proposition}



\section{Functions with three terminals}\Label{S3}
\subsection{Functions for single shot setting}
Now, to deal with joint source and channel coding, we newly introduce some functions related with three random variables $ M, X $ and $ Z $. For $ r>0 $ and 
$ \theta \in (-\infty, 1) $, 
we define following function. 
\begin{align}
U[P_{XZ}, Q_Y; r] (\theta)
:=
r \theta H_{1-\theta}(M) + \theta H_{1-\theta}(P_{XZ}|Q_Y). 
\end{align}
Also we define its derivative
\begin{align}
u[P_{XZ}, Q_Y; r] (\theta)
:=
\frac{d}{d\theta} U[P_{XZ}, Q_Y; r] (\theta).
\end{align}
Since $ U[P_{XZ}, Q_Y; r] (\theta) $ is convex function, $u[P_{XZ}, Q_Y; r] (\theta)$ is monotonically increasing function. Hence, we can define its inverse function 
$\theta[P_{XZ}, Q_Y; r] (a)$ by
\begin{align}
u[P_{XZ}, Q_Y; r] (\theta[P_{XZ}, Q_Y; r](a))
= a, 
\end{align}
for $\underline{a} \le a \le \overline{a}$, 
where $\underline{a}:=\lim_{\theta \rightarrow -\infty} u[P_{XZ}, Q_Y; r] (\theta)$ 
and $\overline{a}:=\lim_{\theta \rightarrow 1} u[P_{XZ}, Q_Y; r] (\theta)$. 

When we define
\begin{align}
R[P_{XZ}, Q_Y; r] (a) := (1 - \theta[P_{XZ}, Q_Y; r] (a))a 
+ U[P_{XZ}, Q_Y; r] (\theta[P_{XZ}, Q_Y; r] (a))
\end{align}
for $\underline{a} \le a \le \overline{a}$, 
the derivative is calculated to be
\begin{align}
\frac{ d R[P_{XZ}, Q_Y; r] (a) }{ da } = (1-\theta(a)). 
\end{align}
Hence, $R[P_{XZ}, Q_Y; r] (a)$ is monotonically increasing function of 
$\underline{a} \le a \le \overline{a}$. Thus, we can define the inverse function $a[P_{XZ}, Q_Y; r] (R)$ by
\begin{align}
R[P_{XZ}, Q_Y; r] (a[P_{XZ}, Q_Y; r] (R)) = R, 
\end{align}
for $R[P_{XZ}, Q_Y; r] (\underline{a}) < R \le r H_0(M) + H_0(X|Z)$. 

\subsection{Functions for two transition matrices}
We define similar functions for two transition matrices $ W_s $ on $ {\cal M} $ and $ W_c $ on $ {\cal X} \times { \cal Z } $. Suppose that $W_c$ is non-hidden with respect to $ {\cal Z} $, i.e., satisfies Assumption 1.

For $ r > 0 $ and $ \theta \in (-\infty, 1) $, we define
\begin{align}
U[W_s, W_c, \downarrow; r] (\theta)
:=& r \theta H_{1-\theta}^{W_s}(M) + \theta H_{1-\theta}^{W_c, \downarrow}(X|Z), \\
u[W_s, W_c, \downarrow; r] (\theta)
:=& \frac{d}{d\theta} U[W_s, W_c, \downarrow; r] (\theta).
\end{align}
Using above two functions, we define
\begin{align}
\theta[W_s, W_c, \downarrow; r](a)
&:= (u[W_s, W_c, \downarrow; r])^{-1} (a), \\
R[W_s, W_c, \downarrow; r] (a) 
&:= (1 - \theta[W_s, W_c, \downarrow; r] (a))a 
+ U[W_s, W_c, \downarrow; r] (\theta[W_s, W_c, \downarrow; r] (a)), 
\end{align}
for $\underline{a} \le a \le \overline{a}$, 
where $\underline{a}:=\lim_{\theta \rightarrow -\infty} u[W_s, W_c, \downarrow; r] (\theta)$ 
and $\overline{a}:=\lim_{\theta \rightarrow 1} u[W_s, W_c, \downarrow; r] (\theta)$. 
Moreover, we define 
\begin{align}
&a[W_s, W_c, \downarrow; r] (R)
:=(R[W_s, W_c, \downarrow; r])^{-1} (R), 
\end{align}
for $R[W_s, W_c, \downarrow; r] (\underline{a}) < R \le r H_0^{W_s} (M) 
+ H_0^{W_c, \downarrow}(X|Z)$.

Now, we suppose that $W_c$ satisfies Assumption 2. 
For $ r > 0 $ and $ \theta, \theta' \in (-\infty, 1) $, we define
\begin{align}
U[W_s, W_c, \theta'; r] (\theta)
:=& r \theta H_{1-\theta}^{W_s}(M) + \theta H_{1-\theta, 1-\theta'}^{W_c}(X|Z), \\
u[W_s, W_c, \theta'; r] (\theta)
:=&\frac{d}{d\theta} U[W_s, W_c, \theta'; r] (\theta).
\end{align}
When $ \theta = \theta' $ we also define 
for $ r > 0 $ and $ \theta \in (-\infty, 1) $, 
\begin{align}
U[W_s, W_c, \uparrow; r] (\theta)
:=& r \theta H_{1-\theta}^{W_s}(M) + \theta H_{1-\theta}^{W_c, \uparrow}(X|Z), \\
u[W_s, W_c, \uparrow; r] (\theta)
:=&\frac{d}{d\theta} U[W_s, W_c, \uparrow; r] (\theta).
\end{align}
Using above two functions, we define 
\begin{align}
\theta[W_s, W_c, \uparrow; r](a)
&:= (u[W_s, W_c, \uparrow; r])^{-1} (a), \\
R[W_s, W_c, \uparrow; r] (a) 
&:= (1 - \theta[W_s, W_c, \uparrow; r] (a))a 
+ U[W_s, W_c, \uparrow; r] (\theta[W_s, W_c, \uparrow; r] (a)), 
\end{align}
for $\underline{a} \le a \le \overline{a}$, 
where $\underline{a}:=\lim_{\theta \rightarrow -\infty} u[W_s, W_c, \uparrow; r] (\theta)$ 
and $\overline{a}:=\lim_{\theta \rightarrow 1} u[W_s, W_c, \uparrow; r] (\theta)$. 
Moreover, we define 
\begin{align}
&a[W_s, W_c, \uparrow; r] (R)
:=(R[W_s, W_c, \uparrow; r])^{-1} (R), 
\end{align}
for $R[W_s, W_c, \uparrow; r] (\underline{a}) < R \le r H_0^{W_s}(M)
 + H_0^{W_c \uparrow}(X|Z)$.

\section{SINGLE SHOT SETTING}\label{S4}

\subsection{Problem formulation}
We first present the problem formulation by the single shot setting. 
Assume that the message $M$ takes values in ${\cal M}$
and is subject to the distribution $P_M$. 
For a channel $W_{Y|X}(y|x)$ with input alphabet ${\cal X}$
and output alphabet ${\cal Y}$, 
a channel code $\phi = (\san{e}, \san{d})$ consists of one encoder 
$\san{e}: {\cal M} \to {\cal X}$ and
one decoder $\san{d}:{\cal Y} \to {\cal M}$. 
The average decoding error probability is defined by
\begin{eqnarray}
\Pjs[\phi|P_M, W_{Y|X}] := \sum_{m \in {\cal M}}
P_M(m) W_{Y|X}(\{b:\san{d}(b) \neq m \}|\san{e}(m)). 
\end{eqnarray}
For notational convenience, we introduce the minimum error probability 
under the above condition:
\begin{eqnarray}
\Pjs(P_M, W_{Y|X}) := \inf_{\phi} \Pjs[\phi|P_M, W_{Y|X}]. 
\end{eqnarray}

\subsection{Direct part}
\subsubsection{General case}
We introduce several lemmas for the case when
${\cal M}$ is the set of messages to be sent, 
$P_M$ is the distribution of the messages, and
$W_{Y|X}$ is the channel from ${\cal X}$ to ${\cal Y}$. 

We have the following single-shot lemma for the direct part. 

\begin{proposition}\cite[Lemma 3.8.1]{han} \Label {si,di,le}
For any constant $c>0$ and for any $P_X \in {\cal P(X)} $, 
there exists a code $ \phi = (\san{e}, \san{d}) $ such that
\begin{align}
\Pjs[\phi|P_M, W_{Y|X}] \le
(P_M \times P_X \times W_{Y|X})
\{ (P_M \times P_X \times W_{Y|X}) (M, X, Y)
\le c (P_X \times \bar{W}_{Y} )(X, Y)
\}+ \frac{1}{c}, \Label{si,di,le1}
\end{align}
where $\bar{W}_{Y}(y):= \sum_{x}P_X(x) W_{Y|X}(y|x)$
and $P_X \times W_{Y|X} (y, x):= P_X (x) W_{Y|X} (y|x)$. 
\end{proposition}

From above Proposition, we obviously have following corollary.
\begin{corollary}
\begin{align}
\Pjs(P_M, W_{Y|X}) \le
(P_M \times P_X \times W_{Y|X})
\{ (P_M \times P_X \times W_{Y|X}) (M, X, Y)
\le c (P_X \times \bar{W}_{Y} )(X, Y)
\}+ \frac{1}{c}. \label{3}
\end{align}
\end{corollary}

\begin{proof}
Since the proof of this lemma is crucial for our proof of the next novel lemma, we give a proof of this lemma as follows.
We prove this lemma by using the random coding method. 
For the code $ \phi= (\san{e}, \san{d})$, we independently choose $\san{e}(m) \in {\cal X}$ subject to $P_X$. 
Define
$D_m:=
\{y| P_M(m) W_{Y|X}(y|\san{e}(m)) \ge c 
\bar{W}_{Y} (y) \}$ 
and define decoding region of message $ m $ as
$D_m':= D_m\setminus (\cup_{m'\neq m}D_{m'})$. 
The error probability of this code can be evaluated as:
\begin{align}
& \Pjs[\phi|P_M, W_{Y|X}] \nonumber \\
\le & 
\sum_{m}P_M(m) 
\big(
W_{Y|X=\san{e}(m)} 
\{P_M(m) W_{Y|X=\san{e}(m)}(Y) < c \bar{W}_{Y} (y) \} \nonumber \\
& \qquad \qquad +\sum_{m'\neq m } W_{Y|X=\san{e}(m)} 
\{P_M(m') W_{Y|X=\san{e} (m')}(Y) \ge c \bar{W}_{Y} (y) \}
\big).\label{si,di3}
\end{align}
Taking the average for the random choice, 
the first term is
\begin{align}
& \mathrm{E}_{\Phi}
\sum_{m}P_M(m) 
W_{Y|X=\san{e}(m)} 
\{P_M(m) W_{Y|X=\san{e}(m)}(Y) < c \bar{W}_{Y} (y) \} \nonumber \\
=&
\sum_{m}P_M(m) 
\sum_{x}P_X(x)
W_{Y|X=x} 
\{P_M(m) W_{Y|X=x}(Y) < c \bar{W}_{Y} (y) \} \nonumber \\
=&
(P_M \times P_X \times W_{Y|X})
\{ (P_M \times P_X \times W_{Y|X}) (M, X, Y)
< c P_X \times \bar{W}_{Y} (X, Y)
\}, \label {si,di1}
\end{align}
and the second term is
\begin{align}
&
\mathrm{E}_{\Phi}
\sum_{m}P_M(m) 
\sum_{m'\neq m } W_{Y|X=\san{e}(m)} 
\{P_M(m') W_{Y|X=\san{e} (m')}(Y) \ge c \bar{W}_{Y} (Y) \}\nonumber \\
=&
\sum_{m, m':m \neq m}
P_M(m) 
\mathrm{E}_{\san{e}(m')}
(\mathrm{E}_{\san{e}(m)} W_{Y|X=\san{e}(m)} )
\{P_M(m') W_{Y|X=\san{e}(m')}(Y) \ge c \bar{W}_{Y} (Y) \} \nonumber \\
=&
\sum_{m, m':m \neq m}
P_M(m) 
\mathrm{E}_{\san{e}(m')}
\bar{W}_{Y} 
\{P_M(m') W_{Y|X=\san{e}(m')}(Y) \ge c \bar{W}_{Y} (Y) \} \label {ss,di}\\
\le &
\sum_{m, m':m \neq m}
P_M(m) 
\mathrm{E}_{\san{e}(m')}
\frac{P_M(m')}{c} W_{Y|X=\san{e}(m')}
\{P_M(m') W_{Y|X=\san{e}(m')}(Y) \ge c \bar{W}_{Y} (Y) \} \nonumber \\
\le &
\sum_{m, m':m \neq m}
P_M(m) 
\frac{P_M(m')}{c} 
\le \frac{1}{c}. \label{si,di2}
\end{align}
Combining (\ref{si,di3}), (\ref{si,di1}) and (\ref{si,di2}), 
we have
\begin{align}
\mathrm{E}_{\Phi}
\Pjs[\phi|P_M, W_{Y|X}]
\le
(P_M \times P_X \times W_{Y|X})
\{ (P_M \times P_X \times W_{Y|X}) (M, X, Y)
\le c (P_X \times \bar{W}_{Y} )(X, Y)
\}+ \frac{1}{c}.
\end{align}
Consequently, there must exist at least one deterministic code 
$ \phi $ satisfying
\begin{align}
\Pjs[\phi|P_M, W_{Y|X}] \le
(P_M \times P_X \times W_{Y|X})
\{ (P_M \times P_X \times W_{Y|X}) (M, X, Y)
\le c (P_X \times \bar{W}_{Y} )(X, Y)
\}+ \frac{1}{c}.
\end{align}

\end{proof}

From the above proof, we also find the following single-shot lemma for the direct part. 

\begin{lemma} \label{si,di,l1}
For any constant $c>0$ and for any distribution $P_X \in {\cal P(X)}$, 
we have
\begin{align}
\Pjs(P_M, W_{Y|X}) 
\le &
(P_M \times P_X \times W_{Y|X})
\{ (P_M \times P_X \times W_{Y|X}) (M, X, Y)
< c P_X \times \bar{W}_{Y} (X, Y)
\}\nonumber \\
&+ 
(1_M \times P_X \times\bar{W}_{Y})
\{ (P_M \times P_X \times W_{Y|X}) (M, X, Y)
\ge c P_X \times \bar{W}_{Y} (X, Y)
\}, \label{4}
\end{align}
where $1_{M}$ is a counting measure on ${\cal M}$. 
The choice $c=1$ gives the minimum upper bound. 
\end{lemma}

We also have following lemma. 
\begin{lemma} \label{si,di,l2}
	\begin{align}
	\Pjs(P_M, W_{Y|X})
	\le 
	e^{s H_{1-s}(M)-s H_{1-s}^{\downarrow}(X|Y) }
	. \label{5}
	\end{align}
\end{lemma}

\begin{proofof}{Lemma \ref{si,di,l1}}
	From (\ref{ss,di}) in the proof of Proposition \ref {si,di,le}, we can evaluate the second term of (\ref {si,di1}) as
	\begin{align*}
	&\sum_{m, m':m \neq m}
	P_M(m) 
	\mathrm{E}_{\san{e}(m')}
	\bar{W}_{Y} 
	\{P_M(m') W_{Y|X=\san{e}(m')}(Y) \ge c \bar{W}_{Y} (Y) \}\\
	&=
	\sum_{m, m':m' \neq m}
	P_M(m) 
	\sum_{ x \in {\cal X} } P_{ X } (x)
	\bar{W}_{Y} 
	\{P_M(m') W_{Y|X=\san{e}(m')}(Y) \ge c \bar{W}_{Y} (Y) \}\\
	&=
	\sum_{m, m':m' \neq m}
	P_M(m) \cdot
	(P_{ X } (x) \times
	\bar{W}_{Y} )
	\{P_M(m') \cdot (P_X \times W_{Y|X=\san{e}(m')})(X,Y) \ge c \bar{W}_{Y} (Y) \}\\
	&\le
	\sum_{m}
	P_M(m) \cdot
	I_M \times P_{ X } (x) \times
	\bar{W}_{Y} 
	\{(P_M(m') \times P_X \times W_{Y|X=\san{e}(m')})(X,Y) \ge c \bar{W}_{Y} (Y) \}\\
	&=
	I_M \times P_{ X } (x) \times
	\bar{W}_{Y} 
	\{(P_M(m') \times P_X \times W_{Y|X=\san{e}(m')})(X,Y) \ge c \bar{W}_{Y} (Y) \}. 
	\end{align*}
So, we obtain (\ref{4}). 

Next, we prove that the right hand side of (\ref{4}) is minimized when $c=1$. 
For any $ c > 0 $, we can evaluate the right hand side of (\ref{4}) as:
\begin{align*}
&
(P_M \times P_X \times W_{Y|X})
\{ (P_M \times P_X \times W_{Y|X}) (M, X, Y)
< c P_X \times \bar{W}_{Y} (X, Y)
\} \\
&+
(1_M \times P_X \times\bar{W}_{Y})
\{ (P_M \times P_X \times W_{Y|X}) (M, X, Y)
\ge c P_X \times \bar{W}_{Y} (X, Y)
\}\\
=& 1-\sum_{(m, x, y): (P_M \times P_X \times W_{Y|X}) (m,x,y)
\ge c P_X \times \bar{W}_{Y} (x,y)}\{(P_M \times P_X \times W_{Y|X})(m, x, y)-(1_M \times P_X \times\bar{W}_{Y})(m, x, y)\}\\
\ge& 1-\sum_{(m, x, y):(P_M \times P_X \times W_{Y|X})(m, x, y)\ge (1_M \times P_X \times\bar{W}_{Y})(m, x, y)}\{(P_M \times P_X \times W_{Y|X})(m, x, y)-(1_M \times P_X \times\bar{W}_{Y})(m, x, y)\}\\
=&(P_M \times P_X \times W_{Y|X})
\{ (P_M \times P_X \times W_{Y|X}) (M, X, Y)
< P_X \times \bar{W}_{Y} (X, Y)
\} \\
&+ 
(1_M \times P_X \times\bar{W}_{Y})
\{ (P_M \times P_X \times W_{Y|X}) (M, X, Y)
\ge P_X \times \bar{W}_{Y} (X, Y)
\}.
\end{align*}
\end{proofof}

\begin{proofof}{Lemma \ref{si,di,l2}}
For any $s \in (0, 1)$, we have
\begin{align*}
&(P_M \times P_X \times W_{Y|X}) 
\{ (P_M \times P_X \times W_{Y|X}) (M, X, Y)
< P_X \times \bar{W}_{Y} (X, Y)
\} \\
&+ 
(1_M \times P_X \times\bar{W}_{Y})
\{ (P_M \times P_X \times W_{Y|X}) (M, X, Y)
\ge P_X \times \bar{W}_{Y} (X, Y)
\}\\
= &\sum_{(P_M \times P_X \times W_{Y|X}) (m, x, y)
< 1_M \times P_X \times \bar{W}_{Y} (m, x, y)} (P_M \times P_X \times W_{Y|X}) (m, x, y) \\
	&+ \sum_{(P_M \times P_X \times W_{Y|X}) (m, x, y)
\ge 1_M \times P_X \times \bar{W}_{Y} (m, x, y)} (1_M \times P_X \times\bar{W}_{Y})(m, x, y)\\
\le &\sum_{(P_M \times P_X \times W_{Y|X}) (m, x, y)< 1_M \times P_X \times \bar{W}_{Y} (m, x, y)}
 	(P_M \times P_X \times W_{Y|X}) (m, x, y)
 \left (
  	\frac{(1_M \times P_X \times\bar{W}_{Y})(m, x, y)}{(P_M \times P_X \times W_{Y|X}) (m, x, y)}
 \right )^{s} \\
&+ \sum_{(P_M \times P_X \times W_{Y|X}) (m, x, y)\ge 1_M \times P_X \times \bar{W}_{Y} (m, x, y)} 
			(1_M \times P_X \times\bar{W}_{Y})(m, x, y)
			\left( \frac{(P_M \times P_X \times W_{Y|X}) (m, x, y)}{(1_M \times P_X \times\bar{W}_{Y})(m, x, y)} 
			\right)^{1-s}\\
=& \sum (P_M \times P_X \times W_{Y|X}) (m, x, y)^{1-s}(1_M \times P_X \times\bar{W}_{Y})(m, x, y)^{s}\\
=& \sum_{m} P_M(m)^{1-s} \sum_{x, y}P_X(x)W_{Y|X} (y)^{1-s} \bar{W}_{Y}(y)^{s} \\
=& 
e^{s H_{1-s}(M)-s H_{1-s}^{\downarrow}(X|Y) }.
\end{align*}

\end{proofof}
However, even when $ M $ is subject to the uniform distribution, the upper bound (\ref{5}) is not so tight. In the uniform case, the Gallager bound is tighter than the upper bound (\ref{5}). So, modifying the derivation of the Gallager bound, we derive joint source and channel coding version of the Gallager bound as follows.

\begin{lemma}
For any distribution $P_X \in {\cal P}({\cal X})$, 
we have
\begin{align}
\Pjs(P_M, W_{Y|X}) 
\le 
e^{\frac{s}{1-s} (H_{1-s}(M)-I_{1-s}^\uparrow (X;Y|P_X \times W_{Y|X}) )}, 
\label{6}
\end{align}
for any $s\in [0, 1/2]$. 
\end{lemma}

\begin{proof}
For encoder, we independently choose $\san{e}(i) \in {\cal X}$ subject to $P_X$, and for decoder, we define decoding region of the message $i$ as
\begin{align}
D(i) := \{ y \in {\cal Y} | \max_{i' \neq i}P_M(i')W_{Y|X=\san{e}(i')}(y) < P_M(i)W_{Y|X=\san{e}(i)}(y) \}. 
\end{align}

And we also define\\
\begin{align}
&\bigtriangleup_{i, j}(y) = \begin{cases}
  0 & P_M(i)W_{Y|X=\san{e}(i)}(y) < P_M(i)W_{Y|X=\san{e}(i)}(y) \\
  1 & P_M(i)W_{Y|X=\san{e}(i)}(y) \ge P_M(i)W_{Y|X=\san{e}(i)}(y),
 \end{cases} \\
&\bigtriangleup_{i, MP}(y) = \begin{cases}
  0 & y \in D(i) \\
  1 & y \notin D(i).
 \end{cases}
 \end{align}
Then, for any $0 \le s \le 1$ and $0 \le t \le 1$, 
\begin{align}
\bigtriangleup_{i, MP}(y) \le \left(\sum_{j}\bigtriangleup_{i, j}(y) \right)^t
\le \left( \sum_{j}\frac{(P_M(j)W_{Y|X=\san{e}(i)}(y))^{1-s}}{(P_M(i)W_{Y|X=\san{e}(i)}(y))^{1-s}} \right)^t, 
\end{align}
and error probability can be represented by
\begin{align}
\Pjs[\phi|P_M, W_{Y|X}] = \sum_{i, y}P_M(i)W_{Y|X=\san{e}(i)}(y)\bigtriangleup_{i, MP}(y). 
\end{align}
So that, 
\begin{align*}
\Pjs[\phi|P_M, W_{Y|X}]
&= \sum_{i, y}P_M(i)W_{Y|X=\san{e}(i)}(y)\bigtriangleup_{i, MP}(y)\\
& \le \sum_{i, y}P_M(i)W_{Y|X=\san{e}(i)}(y)\left( \sum_{j}\frac{(P_M(j)W_{Y|X=\san{e}(i)}(y))^{1-s}}{(P_M(i)W_{Y|X=\san{e}(i)}(y))^{1-s}} \right)^t\\
& \le \sum_{i, y}P_M(i)^{1-t(1-s)}W_{Y|X=\san{e}(i)}(y)^{1-t(1-s)}\left( \sum_{j} (P_M(j)W_{Y|X=\san{e}(i)}(y))^{1-s} \right)^t.
\end{align*}
Taking the average for the random choice, we have
\begin{align}
& E_\Phi \Pjs[\phi|P_M, W_{Y|X}] \nonumber \\
&\le
 \sum_{i, y}P_M(i)^{1-t(1-s)}E_\Phi W_{Y|X=\san{e}(i)}(y)^{1-t(1-s)}\left( \sum_{j} P_M(j)^{1-s} E_\Phi W_{Y|X=\san{e}(i)}(y)^{1-s} \right)^t \nonumber \\
& \le \sum_{i, y}P_M(i)^{1-t(1-s)}\sum_x P_X(x) W_{Y|X}(y|x)^{1-t(1-s)}\left( \sum_{j} P_M(j)^{1-s} \sum_x P_X(x) W_{Y|X}(y|x)^{1-s} \right)^t \nonumber \\
&= \sum_{i}P_M(i)^{1-t(1-s)} \sum_{y}\left( \sum_x P_X(x) W_{Y|X}(y|x)^{1-t(1-s)}\right) \left( \sum_{j} P_M(j)^{1-s}\right)^t \left( \sum_x P_X(x) W_{Y|X}(y|x)^{1-s} \right)^t.
\label{si,di,l3}
\end{align}
By setting $t=\frac{s}{1-s}$ in \eqref{si,di,l3}, we have
\begin{align}
&\sum_{i} P_M(i)^{1-s} \sum_{y} \left( \sum_x P_X(x) W_{Y|X}(y|x)^{1-s}\right) \left( \sum_{j} P_M(j)^{1-s}\right)^{\frac{s}{1-s}} \left( \sum_x P_X(x) W_{Y|X}(y|x)^{1-s} \right)^{\frac{s}{1-s}} \nonumber\\
& =
\left( \sum_{i} P_M(i)^{1-s}\right)^{\frac{1}{1-s}} \sum_{y} \left( \sum_x P_X(x) W_{Y|X}(y|x)^{1-s} \right)^{\frac{1}{1-s}} \nonumber \\
&=
e^{\frac{s}{1-s} (H_{1-s}(M)-I_{1-s}^\uparrow (X;Y|P_X \times W_{Y|X}) )}.
 \label{si,di,l4}
\end{align}
Hence, we have
\begin{align}
 E_\Phi \Pjs[\phi|P_M, W_{Y|X}]
 \ge
e^{\frac{s}{1-s} (H_{1-s}(M)-I_{1-s}^\uparrow (X;Y|P_X \times W_{Y|X}) )}. \label{si,di,l5}
\end{align}
\eqref{si,di,l5} means that 
there must exist at least one deterministic code 
$ \phi $ satisfying 
\begin{align}
\Pjs[\phi|P_M, W_{Y|X}] \ge
e^{\frac{s}{1-s} (H_{1-s}(M)-I_{1-s}^\uparrow (X;Y|P_X \times W_{Y|X}) )}. \label{si,di,l6}
\end{align}
Since $0 \le t \le 1$, $s$ is restricted to $0\le s \le \frac{1}{2}$. 
So we obtain (\ref{6}). 
\end{proof}

\subsubsection{Conditional additive case}
Now, we proceed to the case when the channel is conditional additive. 
Assume that ${\cal X}$ is a module and 
${\cal Y}$ is given as ${\cal X}\times {\cal Z}$. 
Then, the channel $W$ is called conditional additive \cite{HW} when
there exists a joint distribution $P_{XZ}$ such that
\begin{align}
W_{XZ|X}(x, z|x')= P_{XZ}(x-x', z). \label{ca1}
\end{align}
Then we can simplify (\ref{3}). We have following lemma. 
\begin{lemma}
When the channel is conditional additive channel, 
it follows that 
\begin{align}
\Pjs(P_M, W_{XZ|X}) \le
P_M \times P_{XZ} 
\{ P_M(M)P_{X|Z}(X|Z)
\le c \frac{1}{|{\cal X}|}
\} + \frac{1}{c}. \label{a}
\end{align}
\end{lemma}
\begin{proof}
By setting that $P_X$ is the uniform distribution 
and choosing the random variables $X=X'$ and $Y=XZ$ 
to the right hand side of (\ref{3}), 
we have 
\begin{align*}
& (P_M \times P_{X'} \times W_{XZ|X'})
\{ 
(P_M \times P_{X'} \times W_{XZ|X}) (M, X', XZ)
\le c 
P_{X'} \times \bar{W}_{XZ} (X', X, Z)\} \\
= &
(P_M \times P_{X'} \times W_{XZ|X})
\{
P_M(m) \frac{1}{|{\cal X}|} 
P_{XZ}(x-x', z)
\le c 
\frac{1}{|{\cal X}|^2}
P_{Z}(z)
\} \\
= &
(P_M \times P_X \times W_{XZ|X'})
\{
P_M(m)P_{X|Z}(x-x'|z)
\le c 
\frac{1}{|{\cal X}|}
\} \\
= &
P_M \times P_{XZ} 
\{ P_M(M)P_{X|Z}(X|Z)
\le c \frac{1}{|{\cal X}|}
\}, 
\end{align*}
where $ P_Z (z):= \sum_{x} P_{XZ} (x, z) $. 
Hence, 
(\ref{3}) can be simplified to 
\begin{align}
\Pjs(\phi|P_M, W_{Y|X}) \le
P_M \times P_{XZ} 
\{ P_M(M)P_{X|Z}(X|Z)
\le c \frac{1}{|{\cal X}|}
\} + \frac{1}{c}.
\end{align}
\end{proof}

Also we can simplify (\ref{5}) and (\ref{6}). 
We have following lemma. 
\begin{lemma}\label{Lembc}
When the channel is conditional additive channel, 
	it follows that 
\begin{align}
\Pjs(P_M, W_{XZ|X}) \le 
(\frac{e^{H_{1-s}(M)+H_{1-s}^{\downarrow}(X|Z) }}{|{\cal X}|} )^s, 
\label{b}
\end{align}
and
\begin{align}
\Pjs(P_M, W_{XZ|X}) \le 
(\frac{e^{H_{1-s}(M)+H_{1-s}^{\uparrow}(X|Z) }}{|{\cal X}|} )^{\frac{s}{1-s}}. 
\label{c}
\end{align}
\end{lemma}
\begin{proof}
Firstly, we prove \eqref{b}. $ e^{s H^\downarrow (X|Y)} $ is 
represented as:
\begin{align}
& e^{s H_{1-s}^\downarrow (X|Y)}
=\sum_{x, y} 
P_{XY} (x, y)^{1-s} P_Y (y)^s.
\end{align}
Assume that $ {\cal Y} = {\cal X} \times {\cal Z} $ and 
its random variable is $ Y = XZ $. 
Setting $ P_{XY} = P_X \times W_{XZ|X} $, 
$ P_Y (y) = P_Z (z) := \sum_{x} P_{XZ} (x, z) $ 
and 
$ P_X $ is uniform distribution, we have 

\begin{align}
& e^{s H_{1-s}^\downarrow (X|Y)} \nonumber\\
=& 
\sum_{x', x, z} 
\left(
P_X (x') W_{XZ|X} (x, z | x')
\right)^{1-s} P_Z (z)^s \nonumber\\
=&
\sum_{x', x, z} 
\frac{1}{|{\cal X}|^{1-s}}
  P_{XZ}(x-x', z)^{1-s} 
P_{Z} (z)^s \nonumber\\
=&
\left(
\frac{1}{|{\cal X}|}
\right)^{1-s}
\sum_{x} 
e^{s H_{1-s}^\downarrow (X|Z)} \nonumber\\
=&
\frac{e^{s H_{1-s}(X|Z)} }{|{\cal X}|^{s}}. \label{ca,di,l1}
\end{align}
Substituting \eqref{ca,di,l1} to \eqref{5}, 
we have \eqref{b}. 

And also we have 
\begin{align}
&e^{-\frac{s}{1-s}I_{1-s}^\uparrow (X;Y|P_X \times W_{Y|X})} \nonumber\\
=&
e^{-\frac{s}{1-s}I_{1-s}^\uparrow (X;XZ|P_X \times W_{XZ|X})} \nonumber\\
=&
\sum_{x, z} (\sum_{x'} P_X (x') W_{XZ|X}(x,y|x')^{1-s} )^{\frac{1}{1-s}} \nonumber\\
=&
\sum_{x, z} 
(\sum_{x'} \frac{1}{|{\cal X}|} P_{XZ}(x-x', z)^{1-s} )^{\frac{1}{1-s}} \nonumber\\
=& \sum_{x}\frac{1}{|{\cal X}|^{\frac{1}{1-s}}} 
\sum_{z} P_Z(z)
(\sum_{x'} P_{X|Z}(x-x'|z)^{1-s})^{\frac{1}{1-s}} \nonumber \\
=&|{\cal X}|^{1-\frac{1}{1-s}}  \sum_{x}
 e^{\frac{s}{1-s}H_{1-s}^\uparrow(X|Z) } \nonumber \\
=& |{\cal X}|^{-\frac{s}{1-s}} e^{\frac{s}{1-s}H_{1-s}^\uparrow(X|Z) }. \label{ca,di,l2}
\end{align}
Substituting \eqref{ca,di,l2} to \eqref{6}, 
we have \eqref{c}. 
\end{proof}

\subsection{Converse part}
\subsubsection{General case}
Firstly, combining the idea of meta converse \cite{Naga} and \cite[Lemma 4]{HN} and the general converse lemma for the joint source and channel coding \cite[Lemma 3.8.2]{han}, we obtain the following lemma for the single shot setting. 
The following lemma is the same as \cite[Lemma 3.8.2]{han} when $Q_Y$ is $\bar{W}_Y$. 

\begin{lemma}
For any constant $c>0$, any code $ \phi = (\san{e}, \san{d}) $ and any distribution $ Q_Y $ on $ {\cal Y} $, we have
\begin{align}
\Pjs(P_M, W_{Y|X}) \ge
\sum_m P_M(m)W_{Y|X=\san{e}(m)} 
\{ P_M(m)W_{Y|X=\san{e}(m)} (Y)
\le c Q_Y (Y)
\}-c. \label{2}
\end{align}
\end{lemma}
\vspace{2mm}
\begin{proof}
First, we set
\begin{align}
{\cal L} := \{(m, x, y) \in ({\cal M, X, Y}) | P_M(m)W_{Y|X=x} (y)
\le c Q_Y (y)
 \},
\end{align}
and for each $(m, x) \in ({\cal M, X})$, define
\begin{align}
{\cal B}(m, x) := \{ y \in {\cal Y} | (m, x, y) \in {\cal L} \}. 
\end{align}
Also, for decoder $\psi$ and each $m \in {\cal M}$, we define
\begin{align}
{\cal D}(m) := \{y \in {\cal Y} | \psi(y) = m \}. 
\end{align}
In addition, we define $P_{X|M}$ so that
\begin{align}
P_{X|M}(x|m) =
	\begin{cases}
	0 & x \neq \san{e}(m)\\
	1 & x= \san{e}(m).
	\end{cases}
\end{align}
Using this, we define
\begin{align}
P_{MX}(m, x) &:= P_M(m)P_{X|M}(x|m) , \\
P_{MXY}(m, x, y)&:= P_M(m)P_{X|M}(x|m)W_{Y|X=\san{e}(m)} (y). 
\end{align}

Then, 
\begin{align}
&\sum_m P_M(m)W_{Y|X=\san{e}(m)} 
\{ P_M(m)W_{Y|X=\san{e}(m)} (Y)
\le c Q_Y (Y)
\} \nonumber\\
=& \sum_{(m, x, y) \in {\cal L}} P_{MXY} (m, x, y) \nonumber\\
=& \sum_{(m, x) \in {\cal M, X}} \sum_{y \in {\cal B}(m, x)}P_{MX}(m, x) W_{Y|X}(y|x) 
\nonumber\\
=& \sum_{(m, x) \in {\cal M, X}} \sum_{y \in {\cal B}(m, x) \cap {\cal D}(m)}
		P_{MX}(m, x) W_{Y|X}(y|x)
	+\sum_{(m, x) \in {\cal M, X}} \sum_{y \in {\cal B}(m, x)\cap {\cal D}^c(m)}
			P_{MX}(m, x) W_{Y|X}(y|x) \nonumber\\
\le & \sum_{(m, x) \in {\cal M, X}} \sum_{y \in {\cal B}(m, x) \cap {\cal D}(m)}
	P_{MX}(m, x) W_{Y|X}(y|x)
	+\sum_{(m, x) \in {\cal M, X}} \sum_{y \in {\cal D}^c(m)}
	P_{MX}(m, x) W_{Y|X}(y|x) \nonumber\\
=& \sum_{(m, x) \in {\cal M, X}} \sum_{y \in {\cal B}(m, x) \cap {\cal D}(m)}P_{MX}(m, x) W_{Y|X}(y|x) + \Pjs[\phi|P_M, W_{Y|X}]. \label{sbcl1p1}
\end{align}
The last equality follows since the error probability can be written as
\begin{align*}
\Pjs[\phi|P_M, W_{Y|X}] = \sum_{(m, x) \in {\cal M, X}} \sum_{y \in {\cal D}^c(m)}P_{MX}(m, x) W_{Y|X}(y|x). 
\end{align*}
We notice here that
\begin{align*}
 P_M(m)W_{Y|X=\san{e}(m)} (Y)
\le c Q_Y (Y)
\end{align*}
for $y \in {\cal B}(m, x)$. By substituting this into (\ref{sbcl1p1}), the first term of (\ref{sbcl1p1}) is
\begin{align*}
&\sum_{(m, x) \in {\cal M, X}} \sum_{y \in {\cal B}(m, x) \cap {\cal D}(m)}c P_{X|M}(x|m)Q_Y (y)\\
\le & \sum_{(m, x) \in {\cal M, X}} \sum_{y \in {\cal D}(m)}c P_{X|M}(x|m)Q_Y (y)\\
=& c \sum_{m \in {\cal M}} \sum_{y \in {\cal D}(m)}Q_Y (y)\\
=& c \sum_{m \in {\cal M}}Q_Y ({\cal D}(m)) = c, 
\end{align*}
which implies (\ref{2}). 
\end{proof}
\subsubsection{Conditional additive case}
Now, we proceed to the conditional additive case given in \eqref{ca1}. Applying \eqref{2} to the conditional additive case, we obtain following lemma. 
\begin{lemma}\label{sbcl'}
For arbitrary distribution $Q_Z \in {\cal P(Z)}$, we have
\begin{align}
\Pjs(P_M, W_{X, Z|X}) \ge
P_M \times P_{XZ} 
\{ P_M(M)\frac{P_{XZ}(X, Z)}{Q_Z(Z)}
\le c \frac{1}{|{\cal X}|}
\} - c. \label{e}
\end{align}
\end{lemma}
\begin{proof}
For some $Q_Z \in {\cal P}({\cal Z})$, we substitute
\begin{align*}
Q_Y(y) = Q_{XZ}(x, z) = \frac{1}{|{\cal X}|}Q_Z(z)
\end{align*}
to \eqref{2}. Then, the first term of the right hand side of (\ref{e}) is
\begin{align*}
&\sum_m P_M(m)W_{Y|X=\san{e}(m)} 
\{ P_M(m)W_{Y|X=\san{e}(m)} (Y)
\le c Q_Y (Y)
\}\\
=&\sum_m P_M(m)W_{XZ|X=\san{e}(m)} 
\{ P_M(m)W_{XZ|X} (x, z|\san{e}(m))
\le c \frac{1}{|{\cal X}|}Q_Z (z)
\}\\
=&\sum_m P_M(m)P_{XZ} 
\{ P_M(m)P_{XZ} (x-\san{e}(m), y)
\le c \frac{1}{|{\cal X}|}Q_Z (z)
\}\\
=& P_M \times P_{XZ} 
\{ P_M(M)P_{XZ}(X, Z)
\le c \frac{1}{|{\cal X}|}Q_Z (z)
\}. 
\end{align*}
So, we obtain (\ref{e}).

\end{proof}
Similar to \cite[Theorem 5]{HW}, using the monotonicity of R\'{e}nyi divergence, we obtain another type of converse lemma. 
\begin{lemma}\label{sbc} 
We set $ R:=\log|{\cal X}| $. Then, it holds that 
	\begin{align}
	&\log \Pjs[\phi|P_M, W_{Y|X}] \nonumber\\
	\ge&\sup_{s>0, \rho \in \mathbb{R}, \sigma \ge 0}
	\frac{1+s}{s}
	\left[
	-\frac{ U(\rho (1 + s)) }{ 1+s } + U(\rho)  
	+\log
	\left(1-2e^{\frac{U(\rho-\sigma(1-\rho)) - (1+\sigma)U(\rho) + \sigma R}{1+\sigma}}
	\right)
	\right] \label{33}\\
	\ge& \sup_{s>0, \theta (a(R))<\rho<1}
	\frac{1+s}{s}
	\left[	-\frac{ U(\rho (1 + s)) }{ 1+s } + U(\rho) 
	+\log
	\left( 
	1-2e^{(\rho - \theta(a(R))) a(R) + U (\theta(a(R))) - U (\rho)}
	\right)
	\right], \label{sbc1}
	\end{align}
	where
	\begin{align}
U (\cdot) &:= U[P_{XZ}, Q_Z; 1] (\cdot), \label{sbc2}\\
\theta (\cdot)& :=\theta[P_{XZ}, Q_Z; 1] (\cdot), \label{sbc3}\\
a (\cdot)&:= a[P_{XZ}, Q_Z; 1] (\cdot). \label{sbc4}
	\end{align}
\end{lemma}

\begin{proof}
In this proof, we use the notation 
defined in \eqref{sbc2}-\eqref{sbc4}. 

For arbitrary $\rho \in \mathbb{R}$, we define 
following new distributions. 
\begin{align}
P_{M,\rho} (m) &:=
P_M(m)^{1-\rho} e^{ -\rho H_{1-\rho}(M) }, \\
P_{XZ, \rho} (x, z)&:=
P_{XZ, {\rho}}(x,z)^{1-\rho} e^{ -\rho H_{1-\rho}(P_{XZ}|Q_Z) }.
\end{align}
Using these, we define following 
joint distribution. 
\begin{align}
(P_{M} \times P_{XZ, {\rho}})(m, x, z)&:=
P_{M,\rho} (m) P_{XZ, \rho} (x, z) \nonumber \\
&=
(P_M \times P_{XZ})(m, x, z)^{1-\rho}Q_{Z}(z)^{\rho}e^{-U (\rho)}.
\end{align}

For arbitrary code $ \phi = (\san{e}, \san{d}) $, 
we define 
\begin{align}
\alpha := 
\Pjs[\phi|P_M, W_{XZ|X}].
\end{align}
And also, when the source distribution is 
$ P_{M,\rho} $ and 
the channel is conditional additive channel 
$ W_{XZ|X, \rho} $ defined by
\begin{align}
W_{XZ|X, \rho} (x, z| x') :=
P_{XZ, \rho} (x-x', z), 
\end{align}
we define 
\begin{align}
\beta :=
\Pjs[\phi|P_{M,\rho}, W_{XZ|X, \rho}].
\end{align}

Then, for any $s>0$, by the monotonicity of the R\'{e}nyi divergence, we have
\begin{align}
sD_{1+s}(P_M \times P_{XZ, \rho}||P_M \times P_{XZ}) \ge& \log[\beta^{1+s}\alpha^{-s}+(1-\beta)^{1+s}(1-\alpha)^{-s}]\nonumber \\
\ge& \log\beta^{1+s}\alpha^{-s}. 
\end{align}
Thus, we have
\begin{align}
\log\alpha \ge \frac{-sD_{1+s}(P_M \times P_{XZ, \rho}||P_M \times P_{XZ})+(1+s)\log\beta}{s}. \label{sbc9}
\end{align}
For the R\'{e}nyi divergence, we have
\begin{align}
&sD_{1+s}(P_M \times P_{XZ, \rho}||P_M \times P_{XZ})\nonumber\\
=&\log \sum (P_M \times P_{XZ, \rho})^{1+s}(P_M \times P_{XZ})^{-s}\nonumber\\
=& \log \sum (P_M \times P_{XZ})^{1-(1+s)\rho}(Q_{Z})^{(1+s)\rho}e^{(1+s)(U (\rho))}\nonumber\\
=& U ((1+s)\rho) - (1+s) U (\rho). \label{sbc5}
\end{align}

In addition, substituting $ P_M = P_M(m)^{1-\rho} e^{ -\rho H_{1-\rho}(M) }  $ and $ P_{ XZ } =  P_{XZ, {\rho}}(x,z)^{1-\rho} e^{ -\rho H_{1-\rho}(P_{XZ}|Q_Z) }$ into (\ref{e}), we have
\begin{align}
1-\beta \le& (P_M \times P_{XZ, \rho}) 
\{P_M \times P_{XZ, \rho} (m, x, z)
> c \frac{1}{|{\cal X}|}Q_{Z}(z)
\} + c. \label{sbc6}
\end{align}
For any $\sigma \ge 0$, the first term of right hand side of \eqref{sbc6} can 
be evaluated as: 
\begin{align*}
& P_M \times P_{XZ, \rho} 
\{P_M \times P_{XZ, \rho} (m, x, z)
> c \frac{1}{|{\cal X}|}Q_{Z}(z)
\}\\
\le &
\sum_{m,x,z} (P_M \times P_{XZ, \rho})(m, x, z) 
\left(
\frac{P_M \times P_{XZ, \rho}(m, x, z)}{ c \frac{1}{|{\cal X}|}Q_{Z}(z)}
\right)^\sigma\\
=& 
e^{\sigma D(P_M \times P_{XZ, \rho} || Q_{Z}(z)) 
			+ \sigma(\log |{\cal X}| - \log c)  }.
\end{align*}
Thus, by setting $c$ so that
\begin{align}
\sigma D(P_M \times P_{XZ, \rho} || Q_{Z}(z)) 
+ \sigma(\log |{\cal X}| - \log c) 
= \log c, 
\end{align}
we have
\begin{align}
1-\beta \le 
2e^{ \frac{\sigma D(P_M \times P_{XZ, \rho} || Q_{Z}(z)) + \sigma\log|{\cal X}|}
{1 + \sigma}
}. \label{sbc7}
\end{align}
For the R\'{e}nyi divergence in \eqref{sbc7}, we have
\begin{align*}
&\sigma D(P_M \times P_{XZ, \rho} || Q_{Z}(z)) \\
=& 
\log\sum (P_M \times P_{XZ, \rho})(m, x, z)^{1 + \sigma}Q_{Z}(z)^{-\sigma}\\
=&
\log \sum 
\left((P_M \times P_{XZ})(m, x, z)^{1-\rho} 
Q_{Z}(z)^{\rho} e^{-U (\rho)}
\right)^{1 + \sigma} Q_{Z}(z)^{-\sigma}\\
=&
\log \sum 
(P_M \times P_{XZ})(m, x, z)^{1-(\rho - (1 - \rho)\sigma)} 
Q_{Z}(z)^{\rho - (1 - \rho)\sigma} 
-(1 + \sigma)U (\rho) \\
=&
U (\rho - (1 - \rho)\sigma) -(1 + \sigma)U (\rho). 
\end{align*}
So, we have
\begin{align}
\log \beta
\ge
\log
\left(
1-2e^{\frac{U (\rho - (1 - \rho)\sigma) -(1 + \sigma)U (\rho) 
		+ \sigma\log|{\cal X}|}{1+\sigma}}
\right). \label{sbc8}
\end{align}
Combining \eqref{sbc9}, \eqref{sbc5} and \eqref{sbc8}, we obtain (\ref{33}). 

Now, we restrict the range of $\rho$ so that $\theta (a(R))<\rho<1$, and take
\begin{align}
\sigma = \frac{\rho - \theta (a (R))}{1-\rho}, 
\end{align}
we obtain the second inequality. 
\end{proof}

\section{$n$-fold Markovian conditional additive channel}\label{S5}
\subsection{Formulation for general case}\label{s41}
Firstly, we give general notations for channel coding when 
the message obeys Markovian process. 
We assume that the set of messages is ${\cal M}^k$. 
Then, we assume that
the message $M^k=(M_1, \ldots, M_k)\in {\cal M}^k$ 
is subject to the Markov process with 
the transition matrix $\{W_s(m|m')\}_{m, m' \in {\cal M}}$. 
We denote the distribution for $M^k$ by $P_{M^k}$. 

Now, we consider very general sequence of channels with 
the input alphabet ${\cal X}^n$ and the output alphabet ${\cal Y}^n$. 
In this case, the transition matrix as 
$\{W_{Y^n| X^n}(y^n| x^n)\}_{x^n \in {\cal X}^n, y^n \in {\cal Y}^n}$. 
Then, a channel code $\phi = (\san{e}, \san{d})$ consists of one encoder 
$\san{e}: {\cal M}^k \to {\cal X}^n$ and
one decoder $\san{d}:{\cal Y}^n \to {\cal M}^k$. 
Then, the average decoding error probability is defined by
\begin{eqnarray}
\Pj[\phi|k, n|W_s, W_{Y^n|X^n}] := \sum_{m^k \in {\cal M}^k}
P_{M^k}(m^k) W_{Y^n| X^n}
(\{y^n:\san{d}(y^n) \neq m^k \}|\san{e}(m^k)). 
\end{eqnarray}
For notational convenience, we introduce the error probability 
under the above condition:
\begin{eqnarray}
\Pj(k, n|W_s, W_{Y^n|X^n}) := \inf_{\phi} \Pj[\phi|k, n|W_s, W_{Y^n|X^n}]. 
\end{eqnarray}
When there is no possibility for confusion, we simplify it to $\Pj(k, n)$. 
Instead of evaluating the error probability $\Pj(n, k)$ for given $n, k$, 
we are also interested in evaluating 
\begin{eqnarray}
\K(n, \varepsilon|W_s, W_{Y^n|X^n}) := \sup\left\{ k : \Pj(n, k|W_s, W_{Y^n|X^n}) \le \varepsilon \right\}
\end{eqnarray}
for given $0 \le \varepsilon \le 1$.

\subsection{Formulation for Markovian conditional additive channel}
In this section, we address an $n$-fold Markovian conditional additive channel \cite{HW}. 
That is, we consider the case when the joint distribution for the additive noise obeys the Markov process. 
To formulate our channel, we prepare notations. 
Consider the joint Markovian process on ${\cal X}\times {\cal Z}$. 
That is, the random variables $X^n=(X_1, \ldots, X_n)\in {\cal X}^n$ 
and $Z^n=(Z_1, \ldots, Z_n)\in {\cal Z}^n$ 
are assumed to be subject to 
the joint Markovian process defined by 
the transition matrix $\{W_c(x, z|x', z')\}_{x, x'\in {\cal X}, z, z' \in {\cal Z}}$. 
We denote the joint distribution for $X^n$ and $Z^n$ by $P_{X^n, Z^n}$. 
Now, we assume that ${\cal X}$ is a module, and 
consider the channel with
the input alphabet ${\cal X}^n$
and the output alphabet $({\cal X} \times {\cal Z})^n$. 
The transition matrix for the channel 
$W_{X^n, Z^n| \tilde{X}^n}$ is given as
\begin{align}
W_{X^n, Z^n| \tilde{X}^n}(x^n, z^n| \tilde{x}^n)=
P_{X^n, Z^n}(x^n-\tilde{x}^n, z^n )
\end{align}
for $z^n \in {\cal Z}^n$ and $x^n, \tilde{x}^n\in {\cal X}^n$. 
Also, we denote $\log |{\cal X}|$ by $R$. 
In the following discussion, we use the channel capacity 
$C:=\log|{\cal X}| - H^{W_c}(X|Z)$, which is shown in \cite{HW}. 
In this case, we denote 
the average error probability $\Pj [\phi| k, n|W_s, W_{X^n, Z^n|X^n}]$ and
the minimum average error probability $\Pj (k, n|W_s, W_{X^n, Z^n|X^n})$
by $\rom{P_{jca}}[\phi| k, n|W_s, W_c]$ and
$\rom{P_{jca}}(k, n|W_s, W_c)$, respectively. 
Then, we denote the maximum size 
$\K(n, \epsilon|W_s, W_{Y^n|X^n})$
by $\rom{K_{ca}}(n, \epsilon|W_s, W_c)$. 
When we have no possibility for confusion, we simplify them to
by $\rom{P_{jca}}[\phi| k, n]$, $\rom{P_{jca}}(k, n)$, 
and $\rom{K_{ca}}(n, \epsilon)$, respectively. 

In the following discussion, 
we assume Assumption 1 or 2 for
the joint Markovian process described by 
the transition matrix $\{W_c(x, z|x', z')\}_{x, x'\in {\cal X}, z, z' \in {\cal Z}}$. 
The paper \cite{HW} derives the single-letterized channel capacity under Assumption 1. 
Among author's knowledge, 
the class of channels satisfying Assumption 1 is the largest class of channels whose channel capacity is known. 
When ${\cal Z}$ is singleton and the channel is the noiseless channel given by identity transition matrix $I$, 
our problem is the source coding with Markovian source. 
In this case, the memory size is equal to the cardinality $|{\cal X}|^k$, 
we denote the minimum error probability $\rom{P_{jca}}(k, n|W_s, I_{X|X})$ by $\Ps(k, n| W_s)$. 

\subsection{Finite-length bound}

\subsubsection{Assumption 1} 
Now, we assume Assumption 1. 
Combining Proposition \ref{l1} and (\ref{b}) of Lemma \ref{Lembc}, we have an upper bound of the minimum error probability as follows. 

\begin{theorem}[Direct Bound]\label{f1d}
When Assumption 1 holds, setting $ R = \log |{\cal X}| $, we have
\begin{align}
\log \Pj(k, n)
&\le\inf_{s \in (0, 1)}
\left[
-nsR + (n-1)U[ W_s, W_c, \downarrow; \frac{ k-1 }{ n-1 }](s) + \delta(s)
\right], \label{f1d1}
\end{align}
where
\begin{align}
\delta(s) := \overline{\delta}_{W_s}(s) + \overline{\delta}_{W_c}(s). 
\end{align}
\end{theorem}

Combining Proposition \ref{l1} and (\ref{33}) of Lemma \ref{sbc}, we have a lower bound of the minimum error probability as follows. 

\begin{theorem}[Converse bound]\label{f1c}
When Assumption 1 holds, setting $ R = \log |{\cal X}| $, we have
\begin{align}
& \log \Pj(k, n) \nonumber
\\
\ge&
\sup_{s>0, \theta(a(R))<\rho<1} 
\frac{1+s}{s}
\Bigg[	-(n - 1)\frac{ U(\rho (1 + s)) }{ 1+s } + (n - 1)U(\rho) 
+ \delta_{1}(s, \rho) \nonumber \\
 &\qquad \qquad \qquad \qquad \qquad \qquad
+\log
\left( 
1-2e^{(n - 1)((\rho - \theta(a(R))) a(R) + U (\theta(a(R))) - U (\rho)) + \delta_2(\rho)}
\right)
\Bigg], \label{f1c1}
\end{align}
where
\begin{align}
U (\cdot) &:= U[W_s, W_c, \downarrow; \frac{k-1}{n-1}] (\cdot), \\
\theta (\cdot) &:=\theta[W_s, W_c, \downarrow; \frac{k-1}{n-1}] (\cdot), \\
a (\cdot)&:= a[W_s, W_c, \downarrow; \frac{k-1}{n-1}] (\cdot),  
\end{align}
and where
\begin{align}
\delta_1(s, \rho)& := -\frac{\overline{\delta}_{W_s}((1+s)\rho)+\overline{\delta}_{W_c}((1+s)\rho)}{1+s}
+\underline{\delta}_{W_s}(\rho)+\underline{\delta}_{W_c}(\rho), \\
\delta_2(\rho) 
&:= \frac{(1-\rho)(\overline{\delta}_{W_s}((\rho(a(R)))+\overline{\delta}_{W_c}(\rho(a(R))))
	-(1-\rho(a(R)))(\underline{\delta}_{W_s}(\rho)-\underline{\delta}_{W_c}(\rho))
	-(\rho(a(R))-\rho)R}{1-\rho(a(R))}.
\end{align}
\end{theorem}
\begin{proof}
We first substitute $ P_{XZ} = P_{X^nZ^n} $ $Q_{Z}= P_{Z^n}$ to (\ref{33}) of Lemma \ref{sbc} and use Proposition \ref{l1}. 
Then, we restrict the range of $\rho$ as $\theta(a(R)) < \rho < 1$
 and set $\sigma = \frac{\rho - \theta(a(R))}{1-\rho}$. Then, we have the claim 
 of the Theorem. 
\end{proof}

\subsubsection{Assumption 2}
Next, we assume Assumption 2. 
Combining Proposition \ref{l2} and (\ref{c}) of Lemma \ref{Lembc}, we have an upper bound of the minimum error probability as follows. 

\begin{theorem}[Direct Bound]\label{f2d}
When Assumption 2 holds, setting $ R = \log |{\cal X}| $, we have
\begin{align}
\log \Pj(k, n)
\le \inf_{s \in [0, \frac{1}{2}]}
\frac{-nsR+(n-1)U[ W_s, W_c, \uparrow; \frac{ k-1 }{ n-1 }](s)}{1-s}+\xi(s), \label{f2d1}
\end{align}
where
\begin{align}
\xi(s) := \frac{\overline{\xi}_{W_s}(s) + \overline{\xi}_{W_c}(s)}{1-s}. 
\end{align}
\end{theorem}

Combining Proposition \ref{l3} and (\ref{33}), we have a lower bound of the minimum error probability as follows. 

\begin{theorem}[Converse Bound]
\label{f2c}
When Assumption 2 holds, setting $ R = \log |{\cal X}| $, we have
\begin{align}
& \log \Pj(k, n) \nonumber\\
\ge&
\sup_{s>0, \theta(a(R))<\rho<1} 
\frac{1+s}{s}
\Bigg[	-(n - 1)\frac{ U_{\theta(a(R))} (\rho (1 + s)) }{ 1+s } 
+ (n - 1)U_{\theta(a(R))} (\rho) 
+ \delta_{1}(s, \rho)  \nonumber \\
&\qquad \qquad \qquad \qquad \qquad \qquad
+\log
\left( 
1-2e^{(n - 1)((\rho - \theta(a(R))) a(R) + U^\uparrow (\theta(a(R))) 
	- U_{\theta(a(R))} (\rho)) + \delta_2(\rho)}
\right)
\Bigg], \label{g}
\end{align}
where
\begin{align}
&\delta_1 := \underline{\zeta}_{W_c}(\rho, \theta(a(R))) 
- \overline{\zeta}_{W_c}((1+s)\rho, \theta(a(R))), \\
&\delta_2 := \nonumber\\
&\frac{(1-\rho) \{ \overline{\delta}_{W_s}(\theta(a(R)))+\overline{\zeta}_{W_c}(\theta(a(R)), \theta(a(R)))\} - (1-\theta(a(R))) \{ \underline{\delta}_{W_s}(\rho)-\underline{\zeta}_{W_c}(\rho, \theta(a(R)))\} - (\theta(a(R))-\rho) R}{1-\theta(a(R))}, 
\end{align}
and where
\begin{align}
\theta (\cdot)& := [W_s, W_c, \uparrow; \frac{k-1}{n-1}] (\cdot),\\
a (\cdot)&:= a[W_s, W_c, \uparrow; \frac{k-1}{n-1}] (\cdot),\\
U^\uparrow (\cdot) & := U[W_s, W_c, \uparrow; \frac{k-1}{n-1}] (\cdot), \\
U_{\theta(a(R))} (\cdot)& := U[W_s, W_c, \theta(a(R)); \frac{k-1}{n-1}] (\cdot).
\end{align}

\end{theorem}
\begin{proof}
We first substitute $ P_{XZ} = P_{X^nZ^n} $ $Q_{Z}= P_{Z^n}^{(1-\theta(a(R)))}$ to (\ref{33}) of Lemma \ref{sbc} and use Proposition \ref{l2} and \ref{l3}. 
Then, we restrict the range of $\rho$ as $\theta(a(R)) < \rho < 1$
and set $\sigma = \frac{\rho - \theta(a(R))}{1-\rho}$. Then, we have the claim 
of the Theorem. 
\end{proof}

\begin{remark}\label{R1}
Although the paper \cite{VSM} derived a different finite-length converse bound 
as Lemma 3 of \cite{VSM}, their bound contains so large polynomial factor that their bound cannot yield good numerical evaluation as ours.
\end{remark}

\subsection{Large deviation bounds}
In this section, for some constant $r>0$, we fix the coding rate $\frac{k}{n}$ to be $ r $ by using the real number $ R:=\log |{\cal X}| $. 

\subsubsection{Assumption 1}
Now, we assume Assumption 1. Using Theorem \ref{f1d}, we can upper bound the exponent of the minimum error probability as follows. By setting $ k = nr$, taking logarithm and normalizing by $n$ both sides of (\ref{f1d1}), we obtain following theorem.

\begin{theorem}[Direct Bound]\label{T5}
Assume that Assumption 1 holds and 
set $ R = \log |{\cal X}| $. 
When the rate $r$ satisfies 
$rH^{W_{s}}(M) +H^{W_c, \downarrow}(X|Z)<R$, 
we have
\begin{align}
\liminf_{n \rightarrow \infty}- \frac{1}{n}
\log \Pj(n r, n) 
\ge
E_{1,j}(r),\label{ld1d}
\end{align}
where $E_{1,j}(r)$ is error exponent function defined as
\begin{align}
E_{1,j}(r)
:=
\sup_{s \in (0, 1)} [sR  - U[W_s, W_c, \downarrow; r](s) ]. \label{ee1d}
\end{align}
\end{theorem}
\begin{remark}
This theorem is a conditional additive version of \cite[Proposition 1]{ZA}.
\end{remark}

Using Theorem \ref{f1c}, we can lower bound exponent of the minimum error probability as follows. By setting $ k = nr$, we obtain following theorem.

\begin{theorem}[Converse Bound]\label{T6}
Assume that Assumption 1 holds and 
set $ R = \log |{\cal X}| $. 
When the rate $r$ satisfies 
$rH^{W_{s}}(M) +H^{W_c, \downarrow}(X|Z)<R<rH_0^{W_{s}}(M) +H_{0}^{W_c, \downarrow}(X|Z)$, we have
\begin{align}
\limsup_{n \rightarrow \infty} - \frac{1}{n}\log \Pj(rn, n)
\le
\overline{E}_{1,j}(r),
\end{align}
where $\overline{E}_{1,j}(r)$ is error exponent function defined as 
\begin{align}
\overline{E}_{1,j}(r)
:=
& \theta(a(R))a(R) - U[W_s, W_c, \downarrow; r] (\theta(a(R))) \nonumber \\
=& \sup_{\theta \le 1} 
	\frac{
	\theta R - U(\theta)}
	{1-\theta},\label{ee1c}
\end{align}
where
\begin{align}
U (\cdot) := & U[W_s, W_c, \downarrow; r] (\cdot), \\
\theta (\cdot) :=&\theta[W_s, W_c, \downarrow; r] (\cdot), \\
a (\cdot):=& a[W_s, W_c, \downarrow; r] (\cdot).  
\end{align}
\end{theorem}
\begin{remark}
	This theorem is a conditional additive version of \cite[Theorem 2]{ZA}.
\end{remark}
\begin{proof}
From Theorem \ref{f1c}, we have
\begin{align}
\limsup_{n \to \infty} - \frac{ 1 }{ n } \log \Pj(k, n) 
&\le
\frac{1+s}{s}
\left[
\frac{ U(\rho (1 + s)) }{ 1+s } - U(\rho) + \delta_{1}(s, \rho)
\right ]\nonumber\\
&=
\rho 
\frac{ U(\rho (1 + s)) - U(\rho) }{ s \rho } 
- U(\rho) \nonumber\\
&\to
\rho u(\rho) - U(\rho) \quad ({\rm as} \quad s \to 0)\nonumber\\
&\to
\theta(a(R)) u(\theta(a(R))) - U(\theta(a(R))) 
\quad ({\rm as} \quad \rho \to \theta(a(R))) \nonumber\\
&=
\theta(a(R)) a(R) - U(\theta(a(R))), 
\end{align}
where $ u(\cdot) := u[W_s, W_c, \downarrow; r] (\cdot)$. 
\end{proof}

This part will be done similar to \cite[Theorem 21]{HW}. 
In this case, the direct part bound does not coincide with the converse part bound, in general. 
To derive the exact value of the exponent, we need a stronger  assumption.

\subsubsection{Assumption 2}
Next, we assume Assumption 2, which is stronger than Assumption 1. 
Using Theorem 3, we can upper bound the exponent of the minimum error probability as follows. By setting $ k = nr$, taking logarithm and normalizing the both side of (\ref{f2d1}), we obtain following theorem.

\begin{theorem}[Direct Bound]\label{ld2d}
Assume that Assumption 2 holds and 
set $ R = \log |{\cal X}| $. 
When the rate $r$ satisfies 
$rH^{W_{s}}(M) +H^{W_c, \uparrow}(X|Z)<R$, we have
\begin{align}
\liminf_{n \rightarrow \infty}
-\frac{1}{n}\log \Pj(rn, n)
\ge
E_{2,j}(r),
\end{align}
where $E_{2,j}$ is an error exponent function defined as  
\begin{align}
E_{2,j}(r),
:=
\sup_{s \in [0, \frac{1}{2}]}\frac{sR - U[W_s, W_c, \uparrow; r] (s)}{1-s}. \label{ee2d}
\end{align}
\end{theorem}

Using Theorem \ref{f2c}, we can lower bound the exponent of the minimum error probability as follows. By setting $ k = nr$, we obtain following theorem.

\begin{theorem}[Converse Bound]\label{ld2c}
Assume that Assumption 2 holds and 
set $ R = \log |{\cal X}| $. 
When the rate $r$ satisfies 
$rH^{W_{s}}(M) +H^{W_c, \uparrow}(X|Z)<R<rH_0^{W_{s}}(M) +H_{0}^{W_c, \uparrow}(X|Z)$,
 we have
\begin{align}
\limsup_{n \rightarrow \infty}
-\frac{1}{n}\log \Pj(rn, n)
\le \overline{E}_{2,j}(r),
\end{align}
where $\overline{E}_{2,j}(r)$ is an error exponent function defined as
\begin{align}
\overline{E}_{2,j}(r)
:=&
\theta(a(R))a(R) - U^{\uparrow}(\theta(a(R)))\nonumber\\
=& \sup_{0 \le \theta \le 1} 
	\frac{\theta R - U^{\uparrow}(\theta)}{1-\theta}, 
	\label{ee2c}
\end{align}
where
\begin{align}
U^{\uparrow} (\cdot) &:= U[W_s, W_c, \uparrow; r] (\cdot), \\
\theta (\cdot)& :=\theta[W_s, W_c, \uparrow; r] (\cdot), \\
a (\cdot)&:= a[W_s, W_c, \uparrow; r] (\cdot).  
\end{align}
\end{theorem}
\begin{proof}
	From Theorem \ref{f2c}, we have
	\begin{align}
	\limsup_{n \to \infty} - \frac{ 1 }{ n } \log \Pj(k, n) 
	&\le
	\frac{1+s}{s}
	\left[
	\frac{ U_{\theta(a(R))} (\rho (1 + s)) }{ 1+s } 
	- U_{\theta(a(R))} (\rho) + \delta_{1}(s, \rho)
	\right ]\nonumber\\
	&=
	\rho 
	\frac{ U_{\theta(a(R))} (\rho (1 + s)) - U_{\theta(a(R))} (\rho) }{ s \rho } 
	- U(\rho) \nonumber\\
	&\to
	\rho u_{\theta(a(R))} (\rho) - U_{\theta(a(R))} (\rho) \quad ({\rm as} \quad s \to 0)\nonumber\\
	&\to
	\theta(a(R)) u^{\uparrow}(\theta(a(R))) - U(\theta(a(R))) 
	\quad ({\rm as} \quad \rho \to \theta(a(R))) \nonumber\\
	&=
	\theta(a(R)) a(R) - U(\theta(a(R))), 
	\end{align}
	where $ u_{\theta(a(R))} (\cdot) := u[W_s, W_c, \theta(a(R)); r] (\cdot)$ and 
	$ u^{\uparrow}(\cdot) := u[W_s, W_c, \uparrow; r] (\cdot)$. 
\end{proof}
\begin{corollary}
Combining the above theorems, we obtain the exact expression of the exponent of the minimum error probability when we define the critical rate $ R_{cr} $ as

\begin{align}
R_{cr} := R[W_s, W_c, \uparrow; r] \left( u[W_s, W_c, \uparrow; r] 
\left(
\frac{ 1 }{ 2 } 
\right)
\right). 
\end{align}

For $R \le R_{cr}$, we can rewrite the upper bound in Theorem \ref{ld2c} as
\begin{align}
\sup_{s \in [0, \frac{1}{2}]}
\frac{\theta R - U(\theta)}{1-\theta}
= \theta(a(R))a(R) - U(\theta(a(R))).
\end{align}
Thus, the lower bound in Theorem \ref{ld2d} coincides with the upper bound in Theorem \ref{ld2c}. So we have 
\begin{align}
\lim_{n \rightarrow \infty}
-\frac{1}{n}\log \Pj(rn, n) 
&=
\sup_{s \in [0, \frac{1}{2}]}
\frac{\theta R - U(\theta)}{1-\theta} \nonumber\\
&= \theta(a(R))a(R) - U(\theta(a(R))).
\end{align}

\end{corollary}
\begin{remark}
Now, we consider the case when ${\cal Z}$ is singleton and the transition matrix $ W_c $ of the additive noise is the identity matrix $I$, which is the same as the data compression with Markovian source. 
Since $C = \log |{\cal X}|$, we have 
\begin{align}
\lim_{n \to \infty} -\frac{1}{n} 
\log P_{s}(nr, n|W_s) 
\le &
\sup_{s \in (0, 1)} [sR - rsH_{1-s}^{W_{s}}(M)]  \\
\lim_{n \to \infty} -\frac{1}{n} \log P_{s}(nr, n|W_s) 
\ge &
\sup_{\theta \le 1} 
	\frac{\theta R - r \theta H^{W_{s}}_{1-\theta}(M)}{1-\theta}, 
\end{align}
which is the same as the result of \cite[Theorem 12]{HW}.
\end{remark}


\subsection{Moderate deviation bound}
Next, we proceed to the moderate deviation regime, 
in which, the coding rate $ r_n $ behaves as
$r_n := \frac{k}{n}=\frac{C}{H^{W_s}} - \delta n^{-t} $ 
with $ t \in ( 0, \frac{ 1 }{ 2 } ) $. 
Then, the minimum error probability can be evaluated as follows. 

\begin{theorem} \label{md}
Assume that Assumption 1 holds. 
Then, for arbitrary $t \in (0, \frac{1}{2})$ and $\delta > 0$, 
it holds that 
\vspace{2mm}
\begin{align}
\lim_{n \rightarrow \infty} -\frac{1}{n^{1-2t}}
\log \Pj(\frac{nC}{H^{W_s}} - \delta n^{1-t}, n) 
= \frac{1}{2}\cdot\frac{\delta ^{2}}{\frac{1}{(H^{W_{s}}(M))^2}\left [
\frac{C}{H^{W_{s}}(M)}V^{W_{s}}(M)+V^{W_{c}}(X|Z)
\right ]}. \label{md,th}
\end{align}
\end{theorem}
\begin{remark}
Theorem \ref{md} is conditional additive channel version 
of \cite[Theorem 1]{VSM}. 
\end{remark}

\vspace{2mm}
\begin{proof}
From Theorem \ref{f1d}, we obtain
\begin{align}
-\log \Pj(k, n)
&\ge \sup_{s \in (0, 1)}[nsR-(k-1)sH_{1-s}^{W_{s}}(M)-(n-1)sH_{1-s}^{W_c, \downarrow}(X|Z)-\delta(s)]\nonumber\\
&\ge \sup_{s \in (0, 1)}[nsR-(k-1)sH_{1-s}^{W_{s}}(M)-(n-1)sH_{1-s}^{W_c, \downarrow}(X|Z)] + \inf_{s \in (0, 1)}[-\delta(s)]\nonumber\\
&\ge n[s'R-r_ns'H_{1-s'}^{W_{s}}(M)-s'H_{1-s'}^{W_c, \downarrow}(X|Z)] + o(n^{1-2t}). \label{i}
\end{align}
By \eqref{Hdex}, Taylor expansions of $H_{1-s}^{W_{s}}(M)$ and $H_{1-s}^{W_c, \downarrow}(X|Z)$ in the neighborhood of $s=0$ are
\begin{align}
H_{1-s}^{W_{s}}(M) =& H^{W_{s}}(M) + \frac{1}{2}sV^{W_{s}}(M) + o(s), \\
H_{1-s}^{W_c, \downarrow}(X|Z) 
= & H^{W_{c}}(X|Z) + \frac{1}{2}sV^{W_{c}}(X|Z) + o(s).
\end{align}
Substituting these expansions into (\ref{i}), we obtain
\begin{align}
-\log \Pj(\frac{nC}{H^{W_s}} - \delta, n)
\ge n
	\left[
		-\frac{s'^2}{2}(\frac{C}{H^{W_s}}V^{W_{s}}(M) +\right.& V^{W_{c}}(X|Z))
		+s'\delta n^{-t}H^{W_{s}}(M)\nonumber  \\
				-&\left. s'(C+ H^{W_{c}}(X|Z) -\log|{\cal X}|) - \frac{\delta n^{-t}s'^2}{2} + o(s'^2)
				\right]
				+o(n^{1-2t}). 
\end{align}
Now, we set $s':=\frac{\delta n^{-t}H^{W_{s}}(M)}{\frac{C}{H^{W_{s}}(M)}V^{W_{s}}(M)+V^{W_{c}}(X|Z)}$ which satisfies $s \in [0, 1]$ for enough large $n$. Then, we have
\begin{align}
-\log \Pj(k, n)
\ge& n \left[
			\frac{\delta^2 n^{-2t}(H^{W_{s}})^2}{2(\frac{C}{H^{W_{s}}(M)}V^{W_{s}}(M)+V^{W_{c}}(X|Z))} 
			+ o(n^{-2t})
		\right]
		  + o(n^{1-2t})\nonumber\\
=& -n^{1-2t}\frac{1}{2}\cdot\frac{\delta ^{2}}{\frac{1}{(H^{W_{s}}(M))^2}
	\left [
	\frac{C}{H^{W_{s}}(M)}V^{W_{s}}(M)+V^{W_{c}}(X|Z)
	\right ]} + o(n^{1-2t}), 
\end{align}
that is, 
\begin{align}
\liminf_{n \rightarrow \infty}-\frac{1}{n^{1-2t}}\log \Pj(k, n) 
\ge 
\frac{1}{2}
	\cdot\frac{\delta ^{2}}{\frac{1}{(H^{W_{s}}(M))^2}\left [
	\frac{C}{H^{W_{s}}(M)}V^{W_{s}}(M)+V^{W_{c}}(X|Z)
	\right ]}. 
\end{align}

On the other hands, by choosing $\rho = \frac{n^{-t}\delta}{\frac{1}{H^{W_{s}}(M)}\{\frac{C}{H^{W_{s}}(M)}V^{W_{s}}(M)+V^{W_{c}}(X|Z)\}}$, Theorem \ref{f1c} implies that
\begin{align}
& \limsup_{n \rightarrow \infty} -\frac{1}{n^{1-2t}}\log \Pj(k, n) \nonumber\\
\le& \lim_{n \rightarrow \infty} n^{2t} \frac{1+s}{s}\rho [r_n\{H_{1-(1+s)\rho}^{W_{s}}(M)-H_{1-\rho}^{W_{s}}(M)\}
	+(H_{1-(1+s)\rho}^{W_c, \downarrow}(X|Z)-H_{1-\rho}^{W_c, \downarrow}(X|Z))]\nonumber\\
=& \lim_{n \rightarrow \infty} 
	n^{2t} \frac{1+s}{s}\rho
		\frac{1}{2}(r_nV^{W_{s}}(M)+V^{W_{c}}(X|Z)) s\rho \nonumber\\
=& \lim_{n \rightarrow \infty}
	n^{2t} (1+s)\rho^2\frac{1}{2}
			(\frac{C}{H^{W_{s}}(M)}V^{W_{s}}(M)+V^{W_{c}}(X|Z)-\delta n^{-t}V^{W_{s}}(M))\nonumber\\
=& (1+s)\frac{1}{2}\cdot\frac{\delta ^{2}}{\frac{1}{(H^{W_{s}}(M))^2}
	\left [
	\frac{C}{H^{W_{s}}(M)}V^{W_{s}}(M)+V^{W_{c}}(X|Z)
	\right ]}\nonumber\\
\rightarrow& 
\frac{1}{2}\cdot\frac{\delta ^{2}}{\frac{1}{(H^{W_{s}}(M))^2}\left [
	\frac{C}{H^{W_{s}}(M)}V^{W_{s}}(M)+V^{W_{c}}(X|Z)
	\right ]}\qquad\qquad\qquad
(s \rightarrow 0). 
\end{align}
\end{proof}

Now, we consider the case when $ {\cal Z} $ is singleton and the transition matrix $ W_c $ of the additive noise is the identity matrix $I$. 
When $k=\frac{C}{H^{W_s}(M)}n - \frac{C}{H^{W_{s}}(M)^2}(\frac{C}{H^{W_{s}}(M)})^{-t}{\delta}'n^{1-t}$, the minimum error probability $ \Ps(k, n|W_s) $ is characterized as follows. 
Setting $\delta = \frac{C}{H^{W_{s}}(M)^2}(\frac{C}{H^{W_{s}}(M)})^{-t}{\delta}' $ i.e., 
 $k=\frac{C}{H^{W_s}(M)}n - \frac{C}{H^{W_{s}}(M)^2}(\frac{C}{H^{W_{s}}(M)})^{-t}{\delta}'n^{1-t}$, the minimum error probability $ \Ps(k, n|W_s) $ and using $ C=\log |{\cal X}| $, we obtain

\begin{align}
&\lim_{k \rightarrow \infty} -\frac{1}{k^{1-2t}}\log \Pj(k, n) \nonumber\\
=& \lim_{k \rightarrow \infty} -\frac{1}{[ \frac{C}{H^{W_s}(M)}n - C(\frac{C}{H^{W_{s}}(M)})^{-t}{\delta}'n^{1-t}]^{1-2t}}\log \Pj(k, n)\nonumber\\
=& \lim_{k \rightarrow \infty} 
	-\left(\frac{1}{\frac{C}{H^{W_s}(M)}n}\right)^{1-2t} \left(\frac{1}{1-H^{W_s}(M)(\frac{C}{H^{W_{s}}(M)})^{-t}{\delta}'n^{-t}}\right)^{1-2t}\log \Pj(k, n)\nonumber\\
=& \left(\frac{H^{W_s}(M)}{C}\right)^{1-2t} \lim_{k \rightarrow \infty} 
	-\frac{1}{n^{1-2t}}\log \Pj(k, n)\left(\frac{1}{1-H^{W_s}(M)(\frac{C}{H^{W_{s}}(M)})^{-t}{\delta}'n^{-t}}\right)^{1-2t}\nonumber\\
=& \left(\frac{H^{W_s}(M)}{C}\right)^{1-2t}
	\frac{1}{2}\cdot\frac{\delta ^{2}}{\frac{1}{(H^{W_{s}}(M))^2}\left[\frac{C}{H^{W_{s}}(M)}V^{W_{s}}(M)\right]}\nonumber\\
=& \left(\frac{H^{W_s}(M)}{C}\right)^{1-2t}
	\frac{1}{2} \cdot \frac{\frac{C^2}{H^{W_{s}}(M)^4}(\frac{C}{H^{W_{s}}(M)})^{-2t}{\delta'}^2}{\frac{1}{(H^{W_{s}}(M))^2}\cdot \frac{C}{H^{W_{s}}(M)}V^{W_{s}}(M)}\nonumber\\
=&\frac{{\delta'}^2}{2V^{W_{s}}(M)}. 
\end{align}
This result coincides with \cite[Theorem 11] {HW}. 

\section{Numerical Example} \label{S6}
Finally, to demonstrate the advantage of our finite-length bounds,
we numerically evaluate the achievability bound in Theorem \ref{f2d}
and the converse bound in Theorem \ref{f2c}. 
Due to the efficient construction of our bounds, we could calculate both bounds with huge size $n=1\times 10^6$
because the calculation complexity behaves as $O(1)$.
 
We employ the following parametrization $W(p,q)$ for the binary transition matrix:
\begin{eqnarray}
W (p,q):= \left[
\begin{array}{cc}
1-p & q \\
p & 1-q
\end{array}
\right].
\end{eqnarray}
We consider the case when $W_s=W_c=W(0.1,0.2)$.
The optimal transmission rate
$\frac{C}{H^{W_{s}}(M)}$ 
and 
the dispersion ${\frac{1}{(H^{W_{s}}(M))^2}\left [
\frac{C}{H^{W_{s}}(M)}V^{W_{s}}(M)+V^{W_{c}}(X|Z)
\right ]} $
are calculated to be 0.807317 and 6.12809, respectively.
Also, the exponent $E(0.75)$ is calculated to be 0.0002826, which is approximated by
$E_{md}(0.75n,n)/n=  0.0002680$.

When $n=10000$, Fig. \ref{F1} calculates 
the upper and lower bounds of 
$-\log \Pj(k, n)$ based on Theorems \ref{f2d} and \ref{f2c}.
Also, it shows the comparison them with the approximations
$n E(k/n)$ and $E_{md}(k,n)$ by Theorems \ref{ld2d} and \ref{md}.
Fig. \ref{F2} addresses the quantity $-\frac{1}{n} \log \Pj(0.75 n, n)$ in the same way.

\begin{figure}[htbp]
\begin{center}
\scalebox{1}{\includegraphics[scale=1.3]{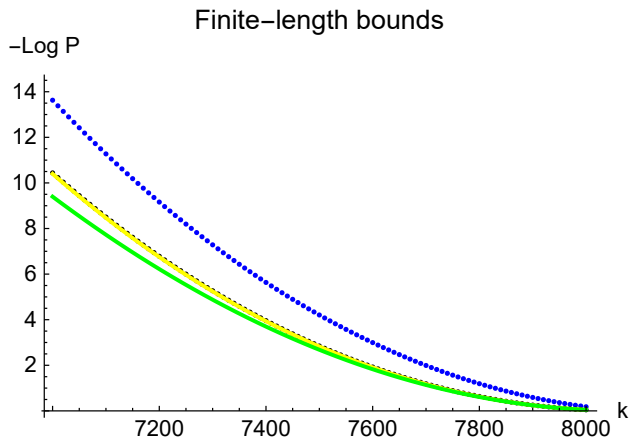}}
\end{center}
\caption{Graphs of 
the upper and lower bounds of 
$-\log \Pj(k, n)$ based on Theorems \ref{f2d} and \ref{f2c} when $n=10000$.
Blue line is the upper bound of $-\log \Pj(k, n)$ based on Theorem \ref{f2c}.
Black line is the lower bound of $-\log \Pj(k, n)$ based on Theorem \ref{f2d}.
Yellow line is $n E(k/n)$.
Green line is $E_{md}(k,n)$.}
\label{F1}
\end{figure}%

\begin{figure}[htbp]
\begin{center}
\scalebox{1}{\includegraphics[scale=1.2]{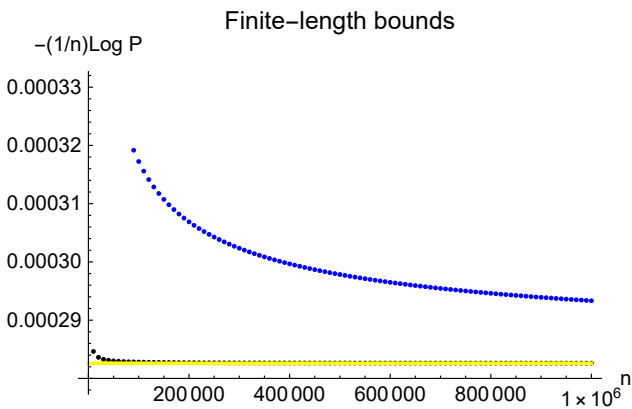}}
\end{center}
\caption{Graphs of 
the upper and lower bounds of 
$-\frac{1}{n} \log \Pj(0.75 n, n)$ based on Theorems \ref{f2d} and \ref{f2c}.
Blue line is the upper bound of $-\frac{1}{n} \log \Pj(0.75 n, n)$ based on Theorem \ref{f2c}.
Black line is the lower bound of $-\frac{1}{n} \log \Pj(0.75 n, n)$ based on Theorem \ref{f2d}.
Yellow line is $E(0.75)$.}
\label{F2}
\end{figure}%


\section*{Acknowledgments}
MH is very grateful to 
Professor Vincent Y. F. Tan and 
Professor Shun Watanabe for helpful discussions and comments.
The works reported here were supported in part by 
a MEXT Grant-in-Aid for Scientific Research (B) No. 16KT0017,
the Okawa Research Grant
and Kayamori Foundation of Informational Science Advancement.

\bibliographystyle{IEEE}

\end{document}